\documentclass[12pt]{article}
\usepackage{euscript,amsmath,amssymb,latexsym,amsthm,enumerate,mathrsfs}
\usepackage[bookmarks=true]{hyperref}

\DeclareMathOperator{\Tr}{Tr}
\renewcommand*{\Re}{\mathop{\mathrm{Re}}\nolimits}

\theoremstyle{remark}
\newtheorem{remark}{Remark}

\theoremstyle{plain}
\newtheorem{theorem}{Theorem}
\newtheorem{proposition}{Proposition}
\newtheorem{corollary}{Corollary}

\begin{document}
\title{
\vspace{1cm} {\bf Semiclassical description of collapses and revivals of quantum wave packets in bounded domains}
}
\author{A.\,S.~Trushechkin$^{1,2}$ and I.\,V. Volovich$^{1}$\bigskip
 \\
{\it  $^{1}$Steklov Mathematical Institute of the Russian Academy of Sciences}
\\ {\it Gubkina St. 8, 119991 Moscow, Russia}\medskip\\
{\it  $^{2}$National Research Nuclear University ``MEPhI''}
\\ {\it Kashirskoe Highway 31, 115409 Moscow, Russia}\bigskip
\\ e-mail:\:\href{mailto:trushechkin@mi.ras.ru}{\texttt{trushechkin@mi.ras.ru}}, \: \href{mailto:volovich@mi.ras.ru}{\texttt{volovich@mi.ras.ru}}}

\date{}
\maketitle

\begin{abstract}
We study a special kind of semiclassical limit of quantum dynamics on a circle and in a box (infinite  potential well with hard walls) as the Planck constant tends to zero and time tends to infinity. The results give detailed information about all stages of evolution of quantum wave packets: semiclassical motion, collapses, revivals, as well as intermediate stages. In particular, we rigorously justify the fact that the spatial distribution of a wave packet is most of the time close to uniform. This fact was previously known only from numerical calculations. 

We apply the obtained results to a problem of classical mechanics: deciding whether recently suggested functional classical mechanics is preferable to traditional Newtonian one from the quantum-mechanical point of view. To do this, we study the semiclassical limit of the Husimi functions of quantum states. We show that functional mechanics remains valid at larger time scales than Newtonian one and, therefore, is preferable. 

Finally, we analyse the quantum dynamics in a box in case when the size of the box is known with a random error. We show that, in this case, the probability distribution of the position of a quantum particle is not almost periodic, but tends to a limit distribution as time indefinitely increases.
\end{abstract}

\section{Introduction}

In this paper we study some aspects of the dynamics of quantum wave packets in bounded domains. These themes are related to fundamental problems of theoretical and mathematical physics. The dynamics of quantum systems in bounded domains has been studied for several decades \cite{1,2,3} and continues to attract attention \cite{4,AronStroud00,5,6,AronStroud05}. Among the important results obtained are an analogue of the Poincar\'{e} recurrence theorem for quantum systems with discrete energy spectrum and a detailed description of the structure of the revival phenomenon in such systems \cite{2,3,AronStroud05}, including the case of the infinite square potential well (box) \cite{4} as well as the case of the finite square potential well \cite{AronStroud00}. Numerical simulation enables one to study in more detail the dynamical properties concerning collapses and revivals of quantum wave packets in bounded domains \cite{5,6}. Understanding quantum dynamics in bounded domains is important in condensed state physics and the physics of nanosystems \cite{7,8}.

There are some open problems in the quantum and classical dynamics of particles in bounded domains. In particular, it is pointed out in \cite{5} that the issue of collapse of quantum wave packets has not yet been adequately studied. Another question: on which time intervals do the quantum and classical descriptions agree? The explicit form of the uncertainty relations for bounded domains is still an open question \cite{9}.
The asymptotic properties of classical dynamics for collisionless continuous media in a box form the subject of papers by Poincar\'{e} and Kozlov (see \cite{10,Kozlov3,11,12}). Here we obtain analogues of the Kozlov's  theorems on diffusion for quantum systems.

We study the dynamics of quantum states on a circle and in a box using a special semiclassical limit as the Planck constant tends to zero and time tends to infinity (a similar procedure is used in the method of the stochastic limit \cite{13}). The results give detailed information about all stages of evolution of quantum wave packets: semiclassical motion, collapses, revivals, as well as intermediate stages. In particular, we rigorously justify the fact (previously known only from numerical calculations) that the spatial distribution of a wave packet is most of the time close to the uniform distribution (an analogue of the Kozlov's theorem on diffusion). This is done in Section~\ref{s2} (the circle case, the main result being Theorem~\ref{PropDynAsymp}) and Section~\ref{s3} (the box case, the main result being Theorem~\ref{PropDynAsympBox}). We prove the theorems for coherent states on a circle and in a box and use the fact that an arbitrary wave function can be represented by an integral over coherent states.

We then apply the  obtained results to a problem of classical mechanics: deciding whether one should prefer  recently suggested functional classical mechanics \cite{14,15} (see also \cite{16,Mikh,30,Pisk,17,18}) to traditional Newtonian one. The basic concept of functional mechanics is not a material point or a trajectory but a probability density function in a phase space. Accordingly, the fundamental dynamical equations are not Newton (or, equivalently, Hamilton) equations but the Liouville equation (even if we consider just one particle, not an ensemble). The Newton (Hamilton) equations become approximate equations for the mean values of distributions of the positions and momenta. Corrections to solutions of the Newton equations have been calculated in some particular cases \cite{14,15,30,Pisk}.

Functional mechanics was suggested in an attempt to solve the irreversibility problem (or reversibility paradox), that is, to make the reversible microscopic dynamics compatible with the irreversible macroscopic dynamics (see \cite{11} as well as \cite{12,19}). This paradox is absent from functional mechanics since both the macro- and microscopic dynamical pictures become irreversible in some sense.

A motivation of functional mechanics comes from the fact that arbitrary real numbers, being infinite decimals, are non-observable (and hence, so are individual trajectories). Therefore, it is more natural to consider bunches of trajectories (or the dynamics of the probability density) than individual trajectories of a material point. Each individual trajectory is a kind of ``hidden variable'' and has no direct physical meaning.

A procedure of constructing the density function of a physical system from directly observable quantities (results of measurements) is described in \cite{17}. The dynamical interaction of the system and the measuring instrument is studied from the point of view of functional mechanics in \cite{18}.

In Section~\ref{s4} we try to approach the problem of choosing a preferable formulation of classical mechanics from a quantum-mechanical perspective. We ask whether Newtonian or functional classical dynamics remains consistent with quantum dynamics for longer time. To answer this, we study the semiclassical limit of the Husimi functions of quantum states of particles on a circle and in a box. For every $\hbar>0$ a quantum density operator determines a classical density function on the phase space, and we pass to the limit as $\hbar\to0$. Note that the evolution of the Wigner function and diffusion in collisionless media consisting of quantum particles on a non-compact space was considered in \cite{20}.

As a result, we obtain (Theorems~\ref{TheoCoherDynClCirc} and \ref{TheoCoherDynClSegm}) that both formulations of classical mechanics adequately describe the system when time is not arbitrarily large. But functional mechanics remains valid at larger time scale than traditional one. Hence, it is preferable in this aspect.

Finally, in Section~\ref{s5} we we analyse the quantum dynamics in a box in case when the size of the box is known with a random error (as we said before, we cannot know the exact size as an infinite decimal). We show that, in this case, the probability distribution of the position of a quantum particle is not almost periodic, but tends to a limit distribution as time indefinitely increases.

\section{Coherent states on a circle}\label{s2}

\subsection{Definition of coherent states on a circle}

Consider a family of functions $\eta_{qp}(x)\in L_2(\mathbb R)$,
$(q,p)\in\mathbb R^2$:
\begin{equation}\label{1} 
\eta_{qp}(x)=\frac{1}{\sqrt[4]{2\pi\alpha^2}}\exp\left\{-\frac{(x-q)^2}{4\alpha^2}+
\frac{ip(x-q)}{\hbar}\right\},
\end{equation}
where $\alpha>0$, $\hbar>0$. It satisfies a property known as the continuous resolution of
unity \cite{21}:
\begin{equation*}
\frac{1}{2\pi\hbar}\iint_{\mathbb R^2}P[\eta_{qp}]\,dqdp=1.
\end{equation*}
Here $P[\psi]$, $\psi\in L_2(\mathbb R)$, stands for the one-dimensional operator acting on any vector $\varphi\in
L_2(\mathbb R)$ by the rule $P[\psi]\varphi=(\psi,\varphi)\psi$, where
$(\cdot,\cdot)$ is the scalar product in $L_2(\mathbb R)$. ($P[\psi]$ is a projector whenever $\psi$ is a unit vector.) The equality is understood in the weak sense: for all $\psi,\chi\in L_2(\mathbb R)$
we have
\begin{equation}\label{2}
\frac{1}{2\pi\hbar}\iint_{\mathbb R^2}(\psi,P[\eta_{qp}]\chi)\,dqdp=
\frac{1}{2\pi\hbar}\iint_{\mathbb
R^2}(\psi,\eta_{qp})(\eta_{qp},\chi)\,dqdp=(\psi,\chi).
\end{equation}

In quantum mechanics, the functions $\eta_{qp}$ (with fixed $\alpha$) are called \textit{coherent states}. The most general formal definition of coherent states was given by Klauder and Skagerstam \cite{22}: a family of coherent states is defined as any family of vectors that continuously depend on their indices and form a resolution of unity. Another key feature of coherent states is that their properties are closest to those of classical particles among all pure quantum states (that is, all square-integrable functions).

The following analogue of the family of coherent states for the spaces $L_2(-l,l)$ was introduced in \cite{23}:
\begin{equation}\label{3}
\upsilon_{qp}(x)=\sum_{n=-\infty}^{+\infty}\eta_{qp}(x-2nl),
\end{equation}
where $q\in[-l,l]$, $p\in\mathbb R$.

\begin{proposition}\label{TheoUnit} The functions (\ref{3}) form a continuous resolution of unity in $L_2(-l, l)$:
$$\frac{1}{2\pi\hbar}\iint_{\Omega}P[\upsilon_{qp}]\,dqdp=1,$$
where $\Omega=\{(q,p)|\,q\in[-l,l],p\in\mathbb R\}$. The equality is understood in the weak sense: for all $\psi,\chi\in L_2(-l,l)$
we have
\begin{equation}\label{4}
\frac{1}{2\pi\hbar}\iint_{\Omega}(\psi,P[\upsilon_{qp}]\chi)\,dqdp=
\frac{1}{2\pi\hbar}\iint_{\Omega}(\psi,\upsilon_{qp})(\upsilon_{qp},\chi)\,dqdp=
(\psi,\chi).
\end{equation}
\end{proposition}

Here $P[\psi]$, $\psi\in L_2(-l,l)$, stands for the one-dimensional operator acting on any vector $\varphi\in L_2(-l,l)$
by the rule $P[\psi]\varphi=(\psi,\varphi)\psi$, where $(\cdot,\cdot)$
stands for the scalar product in $L_2(-l,l)$.
\begin{proof}
 We first establish the simple formula
 \begin{equation*}
(\psi,\eta_{qp}(y-2ml))\equiv\int_{-l}^l\overline\psi(y)\eta_{qp}(y-2ml)dy=
\int_{-\infty}^{+\infty}\overline{\psi_m}(y)\eta_{qp}(y)dy,
\end{equation*}
where
$$\psi_m(x)=\begin{cases}\psi(x+2ml),&x\in[-2ml-l,-2ml+l],\\
0,&x\notin[2ml-l,2ml+l].\end{cases}$$ Using this formula and the property
$\eta_{qp}(x-a)=\eta_{q+a,p}(x)$ for every $a\in\mathbb R$,
we get
\begin{equation*}\begin{split}
\iint_{\Omega}(\psi,\upsilon_{qp})(\upsilon_{qp},\chi)\,dqdp&=
\sum_{n,k=-\infty}^{+\infty}\iint_{\Omega}(\psi,\eta_{qp}(y-2(n+k)l))(\eta_{qp}(x-2nl),\chi)\,dqdp
\\&=\sum_{n,k=-\infty}^{+\infty}\iint_{\Omega}
(\psi_k,\eta_{q+2nl,p})_{\mathbb R}(\eta_{q+2nl,p},\chi_0)_{\mathbb
R}\,dqdp\\&= \sum_{k=-\infty}^{+\infty}\iint_{\mathbb R^2}
(\psi_k,\eta_{qp})_{\mathbb R}(\eta_{qp},\chi_0)_{\mathbb R}\,dqdp.
\end{split}
\end{equation*}
Here $(\cdot,\cdot)_{\mathbb R}$ is the scalar product in
$L_2(\mathbb R)$. Then the desired equality follows from
(\ref{2}):
$$\iint_{\Omega}(\psi,\upsilon_{qp})(\upsilon_{qp},\chi)\,dqdp=
\sum_{k=-\infty}^{+\infty}(\psi_k,\chi_0)_{\mathbb R}=
(\psi_0,\chi_0)_{\mathbb R}=(\psi,\chi).$$
The proposition is proved.
\end{proof}

The functions $\upsilon_{qp}$ may be called \textit{coherent states on a circle} since, first, they satisfy the Klauder--Skagerstam general definition (continuous dependence on the indices and resolution of unity), second, we shall see that the temporal evolution of these states on a circle tends to the dynamics of a classical particle on a circle in the semiclassical limit and, third, they converge to ordinary coherent states $\eta_{qp}$ on a line as $l\to\infty$.

In Section~\ref{s3} we show that the same properties hold for another family of functions in $L_2(-l, l)$:

\begin{equation*}
\omega_{qp}(x)=\sum_{n=-\infty}^{+\infty}(-1)^n\eta_{qp}[(-1)^n(x-2nl)].
\end{equation*}

These are \textit{coherent states in the infinite potential well (in a box)}.

\subsection{Spectral properties of coherent states on a circle} 
One can express the functions $\upsilon_{qp}$ in terms of the theta-function 
\begin{equation*}
\theta(x,\tau)=\sum_{k=-\infty}^{+\infty}\exp\{-\pi\tau k^2+2\pi ikx\}
\end{equation*}
($\Re\tau>0$) as follows:
\begin{equation}\label{5}
\upsilon_{qp}(x)= \frac{1}{\sqrt[4]{2\pi\alpha^2}}\,
\theta\left[\frac{(x-q)l}{2i\pi
\alpha^2}-\frac{pl}{\pi\hbar},\frac{l^2}{\pi\alpha^2}\right]
\exp\left\{-\frac{(x-q)^2}{4\alpha^2}+ \frac{ip(x-q)}{\hbar}\right\}.
\end{equation}

The theta-function has the so-called modular property (Jacobi identity) \cite{24,25}:
\begin{equation*}
\theta(\frac{x}{i\tau},\frac{1}{\tau})=\sqrt\tau e^{\frac{\pi x^2}{\tau}}\theta(x,\tau).
\end{equation*}
Using this identity in (\ref{5}), we obtain after some transformations that

\begin{equation}\label{6}
\begin{split}
\upsilon_{qp}(x)&=\sqrt[4]{\frac{\pi\alpha^2}{2l^4}}\,\theta\left(
-\frac{x-q}{2l}-\frac{p\alpha^2}{il\hbar},\frac{\pi\alpha^2}{l^2}\right)
\exp\left\{-\left(\frac{\alpha p}{\hbar}\right)^2\right\}\\&=
\sqrt[4]{\frac{\pi\alpha^2}{2l^4}}\sum_{k=-\infty}^{+\infty}
\exp\left\{-\alpha^2\left(\frac{\pi}{l}k-\frac{p}{\hbar}\right)^2+i\frac{\pi}{l}k(x-q)\right\}
\end{split}
\end{equation}
We have  got an expansion of $\upsilon_{qp}$ with respect to the orthonormal basis
$$e_k=\frac{1}{\sqrt{2l}}\,e^{i\frac{\pi}{l}kx},\quad k=0,\pm1,\pm2,\ldots,$$ of the space $L_2(-l,l)$. The same result could be obtained by the direct calculation of the scalar product of $\upsilon_{qp}$ with $e_k$. For this purpose, note that
\begin{equation}\label{7} 
\int_{-l}^l\upsilon_{qp}(x)e^{i\frac{\pi}{l}kx}\,dx=
\int_{-\infty}^{+\infty}\eta_{qp}(x)e^{i\frac{\pi}{l}kx}\,dx
\end{equation}
for integer $k$. Indeed,
\begin{multline*}
\int_{-l}^l\upsilon_{qp}(x)e^{i\frac{\pi}{l}kx}\,dx=
\sum_{n=-\infty}^{+\infty}\int_{-l}^l\eta_{qp}(x-2nl)e^{i\frac{\pi}{l}kx}\,dx\\=
\sum_{n=-\infty}^{+\infty}\int_{-2nl-l}^{-2nl+l}\eta_{qp}(x)e^{i\frac{\pi}{l}kx}\,dx=
\int_{-\infty}^{+\infty}\eta_{qp}(x)e^{i\frac{\pi}{l}kx}\,dx.
\end{multline*}

The integral on the right-hand side of (\ref{7}) is a Gaussian integral, whose calculation also yields (\ref{6}). Thus, the modular property of the theta-function turns out to be related to the Fourier transform of a periodic function defined by a sum of Gaussian functions. This is to be expected since the proof of the modular property involves the integration of the Gaussian function against trigonometric functions (see \cite{24} and the more general case in \cite{25}). In Appendix we give a proof of the modular property of the theta-function directly based on the Fourier series expansion of a periodic Gaussian function.

The functions $e_k$ are eigenfunctions of the self-adjoint operator
\begin{equation*}
H^c=-\frac{\hbar^2}{2m}\frac{d^2}{dx^2}\end{equation*}
on $L_2(-l,l)$ with the domain
$$D(H^c)=\{\psi\in AC^2(-l,l)|\,\psi(-l)=\psi(l),\,\psi'(-l)=\psi'(l)\},$$
where $m>0$ is a constant. Here $AC^2(-l,l)$ is the set of differentiable functions whose derivatives lie in $AC(-l, l)$, and $AC(-l, l)$ is the set of absolutely continuous functions whose derivatives lie in $L_2(-l, l)$.
The operator $H^c$ is a Hamiltonian (energy operator) for a free quantum particle of mass $m$ on a circle (see \cite{9}). Thus, the functions $\upsilon_{qp}(x)$ can be expanded in a uniformly convergent series with respect to the eigenfunctions of the Hamiltonian operator of a free particle on a circle.

\begin{remark}\label{r1}By definition, a quantum particle on a circle cannot be free since there must be a potential constraining it to the circle. Only particles on the whole space (on the line in our one-dimensional setting) can be free. But we use the expression ``free quantum particle on a circle'' here to indicate the absence of potentials other than the constraining one. Also, as $l\to\infty$, free motion on a circle becomes free motion on a line.\end{remark}

\subsection{General dynamical properties of the coherent states of a particle on a circle}\label{s23}
The temporal evolution of the state $\upsilon_{qp}$ is described by the formula
$$\upsilon_{qp,t}=U^c_t\upsilon_{qp},$$ where
\begin{equation}\label{8} 
U^c_t=\exp(-\frac{it}{\hbar}H^c)
\end{equation}
is the evolution operator for a free quantum particle on a circle.

The function $\upsilon_{qp,t}$ satisfies the Schr\"{o}dinger equation with the periodic boundary conditions:
\begin{equation}\label{9} 
\begin{split}
&i\hbar\frac{\partial\upsilon_{qp,t}}{\partial t}=
-\frac{\hbar^2}{2m}\frac{\partial^2\upsilon_{qp,t}}{\partial x^2},\\
&\upsilon_{qp,t}(-l)=\upsilon_{qp,t}(l),\,
\upsilon_{qp,t}'(-l)=\upsilon_{qp,t}'(l),\\
&\upsilon_{qp,0}(x)=\upsilon_{qp}(x),
\end{split}
\end{equation}
where $x\in[-l, l]$, $t\in\mathbb R$. Using the reflection method \cite{26}, we get
\begin{equation}\label{10} 
\upsilon_{qp,t}(x)=\sum_{n=-\infty}^{+\infty}\eta_{qp,t}(x-2nl)
\end{equation}
where
\begin{equation}\label{11} 
\eta_{qp,t}(x)=\frac{1}{\sqrt[4]{2\pi\alpha^2(1+i\gamma)^2}}
\exp\left\{-\frac{(x-q-\frac{pt}{m})^2}{4\alpha^2(1+i\gamma)}
+\frac{ip(x-q-\frac{pt}{2m})}{\hbar}\right\},
\end{equation}
is a well-known function in $L_2(\mathbb R)$ that describes the evolution of the initial wave packet $\eta_{qp}$. Here $\gamma=\frac{\hbar t}{2m\alpha^2}.$ We see from (\ref{11}) that under free motion on a line, the centre of the wave packet moves along the classical trajectory $q(t) = q + \frac{pt}m$ while its dispersion grows, that is, the wave packet indefinitely spread with time. If $\Delta q(0) = \alpha$ is the initial mean square deviation of the coordinate, then its value at time $t$ is equal to
\begin{equation}\label{12} 
\Delta q(t)=\sqrt{\alpha^2+\left(\frac{\hbar t}{2m\alpha}\right)^2}.
\end{equation}

On the circle, only the terms with $n = 0$ make an essentially non-zero contribution to the sum (\ref{10}) for small t (we assume that $\alpha\ll l$). Hence, we observe similar behaviour: the centre of the wave packet moves along the classical trajectory with period
\begin{equation*}
T_{cl}=\frac{2lm}p
\end{equation*}
(the period of the motion of a classical particle of mass $m$ and momentum $p$ around the circle of circumference $2l$). The wave packet eventually collapses. The numerical experiments reported in \cite{6} show that at time
\begin{equation}\label{13} 
T_{coll}=\frac{2ml\alpha}{\sqrt3\hbar}\end{equation}
one achieves an approximately uniform distribution of the position of a particle: $|\upsilon_{qp,t}(x)|^2\approx1/{2l}$. The value of $T_{coll}$ is heuristically obtained as
follows. The mean square deviation of the uniform spatial distribution on the circle $[-l,l]$ is equal to $\frac l{\sqrt 3}$. Suppose that $\Delta q(t) = \frac l{\sqrt3}$ , where $\Delta q(t)$ is defined by (\ref{12}). Then $t=T_{coll}$ is a solution of this equation for $t$. Of course, this argument is non-rigorous since $\Delta q(t)$ is the mean square deviation at time $t$ for a particle on a line, not on a circle. But this conclusion is approximately confirmed by numerical calculations. We also obtain some rigorous asymptotic estimates corresponding to the flattening of the spatial density: see the next section. In particular, we will see that the distribution at time $T_{coll}$ (in the semiclassical approach) is non-uniform, the uniform density is achieved at slightly later times. Nevertheless, we will refer to $T_{coll}$ as to a time scale of the wave packet collapse.

The situation at large values of time is quite different from that of dynamics on a line. We express the solution of (\ref{9}) as a series in eigenfunctions of $H^c$:
\begin{equation*}
\upsilon_{qp,t}(x)=\sum_{k=-\infty}^{+\infty}a_{k,qp}\exp\left\{i\frac{\pi}{l}kx-
\frac{i\hbar t}{2m}\left(\frac{\pi}{l}k\right)^2\right\}.
\end{equation*}
The coefficients $a_{k,qp}$ can be found from (\ref{6}). We see that the dynamics is periodic:
the wave packet is completely restored to its original form at the time
\begin{equation}\label{14} 
T_{rev}=\frac{4ml^2}{\pi\hbar}\end{equation}
that is, $\upsilon_{qp,T_{rev}}=\upsilon_{qp}$. This phenomenon is referred to as the \textit{full revival of the wave packet}. At time moments $\frac MN T_{rev}$ with integer $M$ and $N$ one observes the
so-called \textit{fractional revivals of the wave packet} (\cite{2,3}): a copy of the original packet arises simultaneously at several places on the circle. (See \cite{1,2,3,4,5,6} for more details on the dynamics of quantum systems with discrete energy spectrum, including the particular case of quantum particles on a circle.)

Thus, there are three time scales in the quantum dynamics of a particle on a circle \cite{6}:
\begin{enumerate}[1)]
\item $T_{cl}$, the classical period of motion,
\item $T_{coll}$, the characteristic time of collapse of the quantum wave packet,
\item $T_{rev}$, the full revival period of the quantum wave packet.
\end{enumerate}

In the next subsection we consider the semiclassical limit as $\hbar\to 0$, $\alpha \to 0$,
$\frac\hbar\alpha \to 0$ (the parameter $\alpha$ occurs in the definition of $\upsilon_{qp}$; see (\ref{1}) and (\ref{3})). The time scales have different asymptotic behaviour in this limit:
\begin{equation}\label{15} 
T_{cl}=C_1,\quad T_{coll}=C_2\frac{\alpha}\hbar,\quad
T_{rev}=\frac{C_3}\hbar
\end{equation}
where $C_1$, $C_2$, $C_3$ are constants.

\subsection{The semiclassical limit of the dynamics of coherent states on a circle}\label{s24} In this subsection we consider limits in the space of distributions. Let $K$ be a cylinder (the set $\Omega = [-l,l]\times\mathbb R \ni (q,p)$ with the points $(-l,p)$ and $(l,p)$ identified for all $p\in\mathbb R$). We are going to define the space of distributions on $K$. Introduce the set of test functions
\begin{equation}\label{16} 
\begin{split}
\mathscr S(K)=\{\sigma:\mathbb R^2\to\mathbb R|\quad &1)\: \sigma(q+2nl,p)=\sigma(q,p),\: n=0,\pm1,\pm2,\ldots;\\
&2)\: \sigma\in C^\infty(\mathbb R^2);\\
&3) \lim_{p\to\pm\infty}p^r\frac{\partial^{s_1+s_2}\sigma}{\partial
q^{s_1}\partial q^{s_2}}=0,\: r,s_1,s_2=0,1,2,\ldots\}.
\end{split}\end{equation}
Its topology is defined by the seminorms
\begin{equation}\label{17} 
P_N(\sigma)=\max_{s_1+s_2\leq N}\sup_{\mathbb R^2}(1+p^2)^{N/2}
\left|\frac{\partial^{s_1+s_2}\sigma}{\partial q^{s_1}\partial
q^{s_2}}\right|,\quad N=0,1,2,\ldots.\end{equation} 
This is the space of functions on $K$ that decay rapidly with respect to $p$. Let $\mathscr S'(K)$ be the space of distributions, that is, continuous linear functionals
on $\mathscr S(K)$.

Let $\sigma\in\mathscr S(K)$. We define the following distributions:
\begin{align*}
&(\delta(q-q_0,p-p_0),\sigma)=\sigma(q_0,p_0),\\
&(f(q)\delta(p-p_0),\sigma)=\int_{-\infty}^{+\infty}f(q)\sigma(q,p_0),\\
&(c\delta(p-p_0),\sigma)=c\int_{-l}^l\sigma(q,p_0)\,dq,
\end{align*}
where $(q_0,p_0)\in\mathbb R^2$, $f(q)$ is an integrable function on a line, and $c\in\mathbb R$.
Consider the function
\begin{equation}\label{18} 
\varphi_D(q)=\frac1{\sqrt{2\pi
D^2}}e^{-\frac{q^2}{2D^2}},\end{equation} where $D\in(0,\infty)$, and define it for $D = 0$ and $D = \infty$ by putting $\varphi_0(q)=\lim\limits_{D\to0}\varphi_D(q)=\delta(q)$ and
$\varphi_\infty(q)=\lim\limits_{D\to\infty}\varphi_D(q)=\frac1{2l}$
(the limits are taken in the space $\mathscr S'(K)$ of distributions).

The space $\mathscr D(K)$ of test functions is defined by the same formula as $\mathscr S(K)$ but
with the third condition of rapid decay with respect to $p$ in (\ref{16}) replaced by the condition of being compactly supported with respect to $p$. We similarly introduce the space $\mathscr D'(K)$ of distributions.

\begin{theorem}\label{PropDynAsymp}We have the following limit formula in $\mathscr S'(K)$ (where $(q, p)$ are fixed and $(q',p')$ are variables of integration with test functions $\sigma(q',p') \in\mathscr S(K)$):
\begin{equation}\label{19} 
\lim\{\frac{1}{2\pi\hbar}|(\upsilon_{qp},\upsilon_{q'p',t})|^2-
\frac1{N'}\sum_{k=0}^{N'-1}\varphi_D[q'-q-\frac{2kl}{N'}-a+\frac
pm(t-cT_{rev})]\delta(p'-p)\}=0.
\end{equation}
The limit is performed as follows: $\hbar\to0$, $\alpha\to0$,
$\frac\hbar\alpha\to0$, $t=t(\hbar)$, $\frac{\hbar}\alpha(t-c
T_{rev})\to2mD$, $\hbar(t-\frac cT_{rev})\to0$, where $c\in\mathbb R$,
$D\in[0,\infty]$, and the numbers $N'$ and $a$ depend on $c$. If c is rational (and hence, can be written as a reduced fraction $c=\frac MN$), then $N'=N$ for odd $N$ and $N'=\frac N2$ for even $N$. Further, $a=\frac{2l}N$ for $N\equiv 2\pmod 4$ and $a=0$ otherwise. If $c$ is irrational, then $N' = 1$, $a = 0$. The parameter $\alpha$ occurs in the definition of $\upsilon_{qp}$ (see formulae (\ref{1}) and (\ref{3})), $T_{rev}=\frac{4ml^2}{\pi\hbar}$ (see (\ref{14})). The convergence in (\ref{19}) is uniform with respect to $(q, p)\in\Omega$.
\end{theorem}

Let us give some comments on Theorem~\ref{PropDynAsymp}. The quantity $\frac1{2\pi\hbar}|(\upsilon_{qp},\upsilon_{q'p',t})|^2$ is the probability density for a quantum particle to be in the state $\upsilon_{qp}$ at time $t$ under the condition that it was in the state $\upsilon_{q'p'}$ at time 0 (up to the norms $\|\upsilon_{qp}\|^2$ and $\|\upsilon_{q'p'}\|^2$, which tend to 1 in our limit; see Proposition~\ref{LemNormCirc} below).

The limit $\hbar \to 0$ corresponds to the semiclassical approximation. The limits $\alpha \to 0$, $\frac\hbar\alpha\to 0$ correspond to convergence of the mean square deviations of the position and momentum of the quantum wave packet $\upsilon_{qp}$ to zero. We actually have $\Delta q \sim \alpha$
and $\Delta p \sim\frac\hbar{2\alpha}$, where $\Delta q$ and $\Delta p$ are the mean square deviations of the position and momentum (respectively) for $\upsilon_{qp}$ for all $(q,p) \in\Omega$ (see \cite{9}). Here the notation $f \sim g$ means that $\lim \frac fg = 1$. Thus, in the semiclassical limit under consideration, a quantum particle in the state $\upsilon_{qp}$ has a well-defined position (equal to $q$) and
momentum (equal to $p$) just as classical particles do. Therefore, we may say that, in the semiclassical limit, $\frac1{2\pi\hbar}|(\upsilon_{qp},\upsilon_{q'p',t})|^2$ is the probability density for a quantum particle on a circle to be at the phase point $(q, p)$ at time $t$ under the condition that it was at the phase point $(q',p')$ at time $0$.

In Theorem~\ref{PropDynAsymp} we consider various rates of convergence of $t$ to infinity with respect to $\hbar$ and $\alpha$, that is, the various time scales 1)--3) listed in the end of Subsection~\ref{s23}.

We see from (\ref{15}) that the case when $c=0$ and $D=0$ ($\frac{\hbar t}\alpha \to0$) corresponds to the classical time scale $T_{cl}$: time is either fixed or increases slower than the decrease of the collapse velocity (proportional to $\frac\hbar\alpha$) of the packet. Then formula (\ref{19}) takes the form
\begin{equation*}
\lim[\frac{1}{2\pi\hbar}|(\upsilon_{qp},\upsilon_{q'p',t})|^2-
\delta(q'-q+\frac{p}mt,p'-p)]=0.
\end{equation*}
Here $\delta(q' - q + \frac{pt}m, p' - p)$ is the probability density for a classical particle on a circle to be at the phase point $(q, p)$ at time $t$ under the condition that it was at the phase point $(q',p')$ at time 0. Thus, in the semiclassical limit at time scale $T_{cl}$, we have classical dynamics: the quantum probability density of transition to the phase point $(q,p)$ for a particle that was at the phase point $(q',p')$ at time 0 is equal to the corresponding classical probability density.

The case when $c = 0$ and $D \in (0,\infty)$ ($\frac{\hbar t}\alpha \to 2mD \in (0,\infty)$) corresponds to the time scale $T_{coll}$. Formula (\ref{19}) takes the form
\begin{equation*}
\lim[\frac{1}{2\pi\hbar} |(\upsilon_{qp},\upsilon_{q'p',t})|^2-
\varphi_D(q'-q+\frac{p}mt)\,\delta(p'-p)]=0.
\end{equation*}
We see that the quantum probability density of transition from one point to another
is already different from the classical one: there is a spatial spread of probability
distribution. In particular, the case $t = T_{coll}$ corresponds to $D = \frac l{\sqrt 3}$ and, since
$D$ is related to the mean square deviation of the ``spreading function'' $\varphi_D$, this
deviation becomes approximately equal to the mean square deviation of the uniform
distribution on $[-l, l]$, which agrees with the numerical results in \cite{5,6}. However,
the semiclassical limit fails to provide the exact uniform distribution, in contrast to the following case.

The case when $c = 0$ and $D = \infty$ ($\frac{\hbar t}\alpha \to\infty$) corresponds to an intermediate
time scale between $T_{coll}$ and $T_{rev}$. Formula (\ref{19}) takes the form
\begin{equation}\label{20} 
\lim\frac{1}{2\pi\hbar}|(\upsilon_{qp},\upsilon_{q'p',t})|^2=
\frac1{2l}\delta(p'-p),
\end{equation}
so that Theorem~\ref{PropDynAsymp} corresponds to a complete flattening of the spatial probability density in this case. Therefore, we also associate this case with time scale $T_{coll}$ (corresponding to collapse of the localized wave packet). Note that here we get a mathematical justification of the (asymptotic) flattening of the spatial probability density for a quantum particle in a finite volume, which was previously known only from numerical calculations \cite{5,6}. This result may be regarded as a quantum analogue of the Kozlov's  theorems on diffusion for classical systems \cite{Kozlov3,11,12}.

The case $c \neq 0$ corresponds to the time scale $T_{rev}$. If $c$ is irrational, then (just as in the previous case) formula (\ref{19}) reduces to (\ref{20}), that is, one observes a complete flattening of the spatial probability density. The case of rational $c$ corresponds to a revival  (fractional or full) of the wave packet. We discuss this case in more detail.
First suppose that $t - \frac MN T_{rev} \to 0$, whence $D = 0$ (in the simplest case $t = \frac MN T_{rev}$). Then formula (\ref{19}) takes the form
\begin{equation*}
\lim\frac{1}{2\pi\hbar}|(\upsilon_{qp},\upsilon_{q'p',t})|^2=
\frac1{N'}\sum_{k=0}^{N'-1}\delta(q'-q-\frac{2kl}{N'}-a,p'-p),
\end{equation*}
where $N'$ and a are defined in Theorem~\ref{PropDynAsymp}. $N'=1$ (that is, $c$ is integer or half-integer) corresponds to a full revival of the packet, $N'>1$ corresponds to a fractional revival of the packet. This agrees with the results in \cite{2,3,4}. In the case $D \in (0,\infty)$ we obtain a sum of $N'$ ``spread'' wave packets (see (\ref{19}), the formula is not simplified in this case): the revived wave packets begin to spread. In the case $D = \infty$ we again get a complete flattening of the spatial distribution density (\ref{20}).

Since every irrational number can be approximated by rationals within any accuracy, we can (loosely) say that the case of irrational $c$ in Theorem~\ref{PropDynAsymp} is a limiting case of rational $c$ as $N \to\infty$: if we approximate an irrational number by a sequence of rationals, then their denominators increase, the distance between neighbouring terms in the sum of delta-functions in (\ref{19}) tends to zero, and the sum of the delta-functions tends to the uniform distribution (in the weak sense). In other words, the cases of an irrational $c$ and a very close rational $c'$ are almost indistinguishable.

Note that the distribution of the momentum is preserved in all these limiting cases.

Thus, we have traced the whole evolution of a quantum wave packet on a circle. A well-localized initial wave packet eventually collapses until there is a complete flattening of the spatial density. At certain moments we see that copies of the initial packet simultaneously arise at several points of the circle and then again eventually collapse until there is a complete flattening of the density. Since there are ``more'' irrationals than rationals, we can say that the particle most often stays in states whose spatial distribution is close to uniform.

Thus, Theorem~\ref{PropDynAsymp} completely describes the free quantum dynamics of a particle on a circle at all time scales in semiclassical limit. All stages are parametrized by the two real parameters $c$ and $D$.

Here is a simplified version of Theorem~\ref{PropDynAsymp}, which deals only with principal time scales (the classical motion, complete flattening, and exact revivals) without intermediate ones.

\begin{corollary}We have the following limit formulae in $\mathscr S'(K)$:

1) \begin{equation*}
\lim[\frac{1}{2\pi\hbar}|(\upsilon_{qp},\upsilon_{q'p',t})|^2-
\delta(q'-q+\frac{p}mt,p'-p)]=0
\end{equation*}
as $\hbar,\alpha,\frac\hbar\alpha\to0$, $t=const$;

2) $$\lim\frac{1}{2\pi\hbar}|(\upsilon_{qp},\upsilon_{q'p',t})|^2=
\frac1{2l}\delta(p'-p)$$
as $\hbar,\alpha,\frac\hbar\alpha\to0$, $t\to\infty$, $\frac{\hbar t}{\alpha}\to\infty$,\\
as well as $\hbar,\alpha,\frac\hbar\alpha\to0$, $t=cT_{rev}\to\infty$, where $c$ is irrational and $T_{rev}=\frac{16ml^2}{\pi\hbar}$;

3)
\begin{equation*}
\lim[\frac{1}{2\pi\hbar}|(\upsilon_{qp},\upsilon_{q'p',t})|^2
-\frac1{N'}\sum_{k=0}^{N'-1}\delta(q'-q-\frac{2kl}{N'}-a,p'-p)]=0
\end{equation*}
as $\hbar,\alpha,\frac\hbar\alpha\to0$, $t=cT_{rev}\to\infty$, where $c=\frac MN$ is rational and the numbers $N'$ and $a$ depend on $N$.

These limits are uniform with respect to $(q,p)\in \Omega$.
\end{corollary}

\begin{proof}[Proof of Theorem~\ref{PropDynAsymp}]
Let us express the scalar product of $\upsilon_{qp}$ and
$\upsilon_{q'p',t}$ in $L_2(-l,l)$ in terms of the scalar products of $\eta_{qp}$ and $\eta_{q'+2kl,p',t}$,
$k=0,\pm1,\pm2,\ldots$, in $L_2(\mathbb R)$:
\begin{equation}\label{21} 
\begin{split}
(\upsilon_{qp},\upsilon_{q'p',t})&=\sum_{n,k=-\infty}^{+\infty}
\int_{-l}^l\overline{\eta_{qp}}(x-2kl)\eta_{q'p',t}(x-2(k+n)l)\,dx\\&=
\sum_{k=-\infty}^{+\infty}\int_{-\infty}^{+\infty}
\overline{\eta_{qp}}(x)\eta_{q'+2kl,p',t}(x)\,dx=
\sum_{k=-\infty}^{+\infty}(\eta_{qp},\eta_{q'+2kl,p',t}).
\end{split}\end{equation}

Substituting in (\ref{21}) the expression for the scalar product
\begin{multline}\label{22} 
(\eta_{qp},\eta_{q'p',t})=\sqrt{\frac{2}{2+i\gamma}}
\exp\left\{-\frac{(q'-q+\frac{(p'+p)t}{2m})^2}{4\alpha^2(2+i\gamma)}-
\frac{\alpha^2(p'-p)^2}{2\hbar^2}-\right.\\\left.-
\frac{i(p'+p)(q'-q)}{2\hbar}-\frac{it(p'+p)^2}{8m\hbar}\right\},
\end{multline}
where $\gamma=\frac{\hbar t}{2m\alpha^2}$ as above, we get
\begin{multline*}
(\upsilon_{qp},\upsilon_{q'p',t})= \sqrt{\frac{2}{2+i\gamma}}
\sum_{k=-\infty}^{+\infty} \exp\left\{
-\frac{(q'-q+2kl+\frac{(p'+p)t}{2m})^2}{4\alpha^2(2+i\gamma)}-
\frac{\alpha^2(p'-p)^2}{2\hbar^2}-\right.\\\left.-
\frac{i(p'+p)(q'-q+2kl)}{2\hbar}-\frac{it(p'+p)^2}{8m\hbar} \right\},
\end{multline*}
\begin{multline*}
|(\upsilon_{qp},\upsilon_{q'p',t})|^2= \frac{2}{\sqrt{4+\gamma^2}}
\sum_{k,n=-\infty}^{+\infty} \exp\left\{
-\frac{(q'-q+2kl+\frac{(p'+p)t}{2m})^2}{4\alpha^2(2+i\gamma)}-
\right.\\\left.
-\frac{(q'-q+2nl+\frac{(p'+p)t}{2m})^2}{4\alpha^2(2-i\gamma)}-
\frac{\alpha^2(p'-p)^2}{\hbar^2}- \frac{i(p'+p)(k-n)l}{\hbar}
\right\}.
\end{multline*}
Put $k=n+r$ and replace the sum over $k$ and $n$ by a sum over $r$ and $n$:
\begin{multline*}
|(\upsilon_{qp},\upsilon_{q'p',t})|^2= \frac{2}{\sqrt{4+\gamma^2}}
\sum_{r,n=-\infty}^{+\infty} \exp\left\{
-\frac{(q'-q+2nl+rl+\frac{(p'+p)t}{2m})^2}{\alpha^2(4+\gamma^2)}
-\right.\\\left. -\frac{r^2l^2}{\alpha^2(4+\gamma^2)}
+\frac{irl\gamma(q'-q+2nl+rl+\frac{(p'+p)t}{2m})}{\alpha^2(4+\gamma^2)}
-\frac{\alpha^2(p'-p)^2}{\hbar^2}- \frac{i(p'+p)rl}{\hbar} \right\}.
\end{multline*}

Let $\sigma(q',p')$ be an arbitrary test function from $\mathscr
S(K)$. We expand it in a Fourier series with respect to $q'$ and represent it by a Fourier integral with respect to $p'$:
\begin{equation}\label{23} 
\sigma(q',p')=\frac1{2\sqrt{\pi
l}}\sum_{j=-\infty}^{+\infty}\int_{-\infty}^{+\infty}
a_j(\nu)\exp\left\{i\frac{\pi}{l}jq'+i\nu p'\right\} d\nu.\end{equation} 
We calculate the following integral:
\begin{multline}\label{24} 
\iint_\Omega |(\upsilon_{qp},\upsilon_{q'p',t})|^2
\exp\left\{i\frac{\pi}{l}jq'+i\nu p'\right\}\,dq'dp'\\= \iint_\Omega
|(\upsilon_{qp},\upsilon_{q'-\frac{(p'+p)t}{2m},p',t})|^2
\exp\left\{i\frac{\pi}{l}j\left(q'-\frac{pt}{2m}\right)+ip'\left(\nu-\frac{\pi jt}{2lm}\right)\right\}\,dq'dp'\\=
\frac{2}{\sqrt{4+\gamma^2}} \sum_{r=-\infty}^{+\infty}
\iint_{\mathbb R^2} \exp\left\{
-\frac{(q'-q+rl)^2}{\alpha^2(4+\gamma^2)}
-\frac{r^2l^2}{\alpha^2(4+\gamma^2)}
+\frac{irl\gamma(q'-q+rl)}{\alpha^2(4+\gamma^2)} -\right.\\\left.
-\frac{\alpha^2(p'-p)^2}{\hbar^2} -\frac{i(p'+p)rl}{\hbar}
+i\frac\pi lj\left(q'-\frac{pt}{2m}\right) +ip'\left(\nu-\frac{\pi jt}{2lm}\right)
\right\}\,dq'dp'\\= 2\pi\hbar \sum_{r=-\infty}^{+\infty}
\exp\left\lbrace -\frac{\alpha^2(4+\gamma^2)}4
\left[\frac{rl\gamma}{\alpha^2(4+\gamma^2)}+\frac{\pi j}l\right]^2
-\frac1{4\alpha^2}\left[rl+\frac{\pi j\hbar
t}{2lm}-\nu\hbar\right]^2 -\right.\\\left.
-\frac{2iprl}\hbar-\frac{r^2l^2}{\alpha^2(4+\gamma^2)} +i\frac\pi
lj\left(q-rl-\frac{pt}m\right)+ip\nu \right\rbrace.
\end{multline}

We now realize all the passages to the limit. First consider the case when
$c=0$, $D\in(0,\infty)$. Only the term with $r=0$ remains non-zero in (\ref{24}). We have
\begin{multline*}
\lim\left[\frac1{2\pi\hbar}\iint_\Omega |(\upsilon_{qp},\upsilon_{q'p',t})|^2
\exp\left\{i\frac{\pi}{l}jq'+i\nu p'\right\}\,dq'dp'\right.\\\left.-\exp\left\{-\frac12\left(\frac{\pi j}{l}D\right)^2
+i\frac\pi lj\left(q-\frac pmt\right)+ip\nu\right\}\right]=0.
\end{multline*}

This means that
\begin{multline*}
\lim\left[\frac1{2\pi\hbar}
\iint_\Omega|(\upsilon_{qp},\upsilon_{q'p',t})|^2\sigma(q',p')\,dq'dp'\right.\\\left.-
\frac{1}{2\sqrt{\pi
l}}\sum_{j=-\infty}^{+\infty}\int_{-\infty}^{+\infty}
a_j(\nu)\exp\left\{-\frac12\left(\frac{\pi j}{l}D\right)^2+i\frac\pi
lj\left(q-\frac{pt}m\right)+ip\nu\right\}d\nu\right]=0.\end{multline*}
Clearly, the integrals and the series converge uniformly in $(q,p)\in\Omega$. Again expressing $a_j(\nu)$ in terms of $\sigma(q',p')$
by the formula
$$a_j(\nu)=\frac1{2\sqrt{\pi l}}\iint_\Omega\sigma(q',p')
\exp\left\{-i\frac\pi ljq'-i\nu p'\right\}dq'dp'$$ and using the modular property of the theta-function, we obtain (\ref{19}). A similar and simpler argument proves (\ref{19}) for $D=0$ and $D=\infty$
(here $c=0$ as above) as well as for irrational $c$.

Consider the case of an arbitrary rational $c=\frac MN$. We first assume for simplicity that $t-\frac MNT_{rev}\to0$ (whence $D=0$). Then the terms of
(\ref{24}) with non-zero limits are only those with
$r=-2cj$ (this can be seen from the two first terms in the exponent). Accordingly, the terms of the sum (\ref{23}) have zero limits unless $j$ is such that $2cj$ is an integer.
Namely, the terms with non-zero limits are those with $j=N'J$,
$J=0,\pm1,\pm2,\ldots$, where $N'=N$ for odd $N$ and $N'=\frac
N2$ for even $N$. We look at the term $-i\pi jr=2\pi
icj^2$ in the exponent of (\ref{24}) in more detail. If $N$ is odd, then
$\exp(2\pi icj^2)=\exp(2\pi iMNJ^2)=1$ for all $J$.
If $N$ is even, then $\exp(2\pi icj^2)=\exp(i\pi
MN'J^2)=(-1)^J=\exp(i\pi N'J)$ (if $N$ is even, $M$ must be odd since $M$ and $N$ are coprime). Thus, for odd $N$ we get
\begin{equation}\label{25} 
\lim\frac1{2\pi\hbar}
\iint_\Omega|(\upsilon_{qp},\upsilon_{q'p',t})|^2\sigma(q',p')\,dq'dp'=
\frac1{2\sqrt{\pi
l}}\sum_{J=-\infty}^{+\infty}\int_{-\infty}^{+\infty}
a_{NJ}(\nu)\exp\left\{i\frac\pi l NJq+ip\nu\right\}.\end{equation} This is equivalent to saying that
$$\lim\frac1{2\pi\hbar}
\iint_\Omega|(\upsilon_{qp},\upsilon_{q'p',t})|^2\sigma(q',p')\,dq'dp'=
\frac1{N}\sum_{k=0}^{N-1}\sigma(q+\frac{2kl}N,p).$$ This proves formula
(\ref{19}) for odd $N$. If $N$ is even, we get
\begin{multline}\label{26} 
\lim\frac1{2\pi\hbar}
\iint_\Omega|(\upsilon_{qp},\upsilon_{q'p',t})|^2\sigma(q',p')\,dq'dp'\\=
\frac1{2\sqrt{\pi
l}}\sum_{J=-\infty}^{+\infty}\int_{-\infty}^{+\infty}
a_{N'J}(\nu)\exp\left\{i\frac\pi l N'J(q+l)+ip\nu\right\}.\end{multline} This is equivalent to the equation
$$\lim\frac1{2\pi\hbar}
\iint_\Omega|(\upsilon_{qp},\upsilon_{q'p',t})|^2\sigma(q',p')\,dq'dp'=
\frac1{N'}\sum_{k=0}^{N'-1}\sigma(q+l+\frac{2kl}{N'},p).$$ If $N'$
is even (so that $N$ is divisible by 4), then $l+\frac{2kl}{N'}=0$ for some
 $k$. If $N'$ is odd (so that $N\equiv2\pmod 4$), then
$l+\frac{2kl}{N'}=\frac{2l}N$ for some $k$. Therefore, we can write
$$\lim\frac1{2\pi\hbar}
\iint_\Omega|(\upsilon_{qp},\upsilon_{q'p',t})|^2\sigma(q',p')\,dq'dp'=
\frac1{N'}\sum_{k=0}^{N'-1}\sigma(q+\frac{2kl}{N'}+a,p),$$ where $a=0$
if $N$ is divisible by 4 and $a=\frac{2l}N$ if $N\equiv 2\pmod 4$. This proves formula (\ref{19}) for even $N$.

The same argument works when the condition $t-\frac MNT_{rev}\to0$ does not hold, but still we have
$D=0$. The only difference is that the first argument of $\sigma$ acquires an additional summand
$\frac pm(t-\frac MNT_{rev})$ in the final expression (because the sum
$\frac{2iprl}\hbar+\frac{i\pi jpt}{lm}$ in the exponent of 
(\ref{24}) does not tend to zero in this case).

We similarly treat the case when $D\in(0,\infty)$. Again, the terms of
(\ref{24}) (resp. of the sum (\ref{23})) have zero limits unless
$r=-2cj$ (resp. $j$ is such that $2cj$ is an integer). However, just as in the case when $c=0$, the integrands of the right-hand sides of (\ref{25}) and (\ref{26}) acquire a factor $e^{-\frac12(\frac{\pi j}lD)^2}$, which results in the replacement
of the delta-functions by the functions $\varphi_D$. In the case when $D=\infty$, all terms of the sum over $k$ tend to $\frac1{2l}\delta(p'-p)$, which yields formula (\ref{19}) for the last limiting case.

The theorem is proved.
\end{proof}
We now prove another proposition to be used in what follows.

\begin{proposition}\label{LemNormCirc}The norm of $\upsilon_{qp}$ tends to unity uniformly on $\Omega$ as $\hbar \to 0$, $\alpha \to 0$, $\frac\hbar\alpha \to0$.
\end{proposition}

\begin{proof} By formula (\ref{21}) we have
\begin{equation}\label{27} 
\|\upsilon_{qp}\|^2=
\sum_{k=-\infty}^{+\infty}\exp\left\{-\frac{l^2k^2}{2\alpha^2}+
\frac{2iplk}{\hbar}\right\}\to1.
\end{equation}
\end{proof}

Here are some remarks on Theorem~\ref{PropDynAsymp}.

\begin{remark}We refer to the limit performed as to the \textit{special} semiclassical limit. The usual well-known semiclassical limit (see, for example, \cite{27}) is the limit $\hbar\to0$ with constant $t$. In our case, this corresponds to the first time scale $T_{cl}$ (classical motion). We consider the simultaneous limits $\hbar\to0$ and $t\to\infty$, where $t$ and $\hbar$ are related to each other and to $\alpha\to0$ in certain different ways. This allowed us to investigate analytically not only the first (classical) time scale, but the other two time scales and the intermediate time scales as well.
\end{remark}

\begin{remark}\label{r2}We have proved the theorem on the semiclassical limit of dynamics on a circle only for the quantum states of special form. But the result obtained can be extended to more general states since, by formula (\ref{4}), every function $\psi \in L_2(-l,l)$ can be represented by an integral over coherent states:
$$\psi=\frac1{2\pi\hbar}\iint_\Omega f_\psi(q,p)\upsilon_{qp}\,dqdp,$$
where $f_\psi(q,p)=(\upsilon_{qp},\psi)$.


Nevertheless, it is worthwhile to give a formulation and a proof of Theorem~\ref{PropDynAsymp} which would not distinguish any particular form of wave packets at all.\footnote{The authors are grateful to J.\,R.~Klauder for this important remark.} This is a subject for further work.
\end{remark}

\begin{remark}\label{r3}Consider the dynamics of the mean position and mean momentum of the
state $\upsilon_{qp,t}$. The momentum is easily seen to be preserved under our semiclassical
limit $\hbar, \alpha, \frac\hbar\alpha \to 0$ at all time scales, that is, for any behaviour of the time variable $t$:
$$\overline
p=-i\hbar\int_{-l}^{l}\overline{\upsilon_{qp,t}(x)}\upsilon'_{qp,t}(x)\,dx\to
p.$$

We now look at the dynamics of the mean position. Since the position operator is
not well-defined for a particle on a circle, one usually considers the mean complex exponent of the position, which is uniquely determined:
$$\left(\overline{e^{i\frac\pi lq}}\right)_t=
\int_{-l}^le^{i\frac\pi lx}|\upsilon_{qp,t}(x)|^2\,dx.$$
Consider the time scale $T_{cl}$, that is, $\frac{\hbar t}\alpha\to0$. We assume for simplicity that $t=const$. Since
$|\upsilon_{qp,t}(x)|^2$ tends to the periodic delta-function
$\sum_n\delta(x-q-\frac pmt-2nl)$ at this time scale, we have
$$\left(\overline{e^{i\frac\pi lq}}\right)_t\to e^{i\frac\pi l(q+\frac pmt)}.$$
Thus, at time scale $T_{cl}$, the centre of the wave packet moves along the classical trajectory. At other times scales, the notions of ``centre of the packet'' and ``mean position'' have no physical meaning since the packet either collapses and has no definite centre or has several centres (fractional revivals). In the last case the usual expectation of the (complex exponent of the) position may lie between these centres. This also lacks physical meaning because the particle cannot be observed near this expected value. Only a full revival makes the notion of position meaningful again.
\end{remark}

\begin{remark}\label{r4}We use the limit $\hbar,\alpha,\frac\hbar\alpha\to 0$ as a mathematical tool. In particular, it allows us to formulate the result about the flattening of the spatial density. As already mentioned, the fact that the spatial density of a particle in a finite volume is most of the time close to uniform, has not yet been proved in the literature nor even stated in a mathematically rigorous manner: the uniform distribution is never achieved exactly for a localized initial packet. What then is the meaning of the ``closeness'' of our distribution to the uniform one? We have proved that the spatial density is exactly uniform in our limiting case at certain time scales (scale $T_{coll}$ and most of scale $T_{rev}$). This is a mathematical expression of the fact that the spatial density is close to uniform.

These limits have no direct meaning from the physical point of view: Planck constant $\hbar$ is a physical constant and cannot tend to zero (see \cite{28} for a discussion
of the semiclassical limit with the constant $\hbar$ for the baker's map). The parameter $\alpha$ is also
fixed for a fixed coherent state. From a physical point of view, it would be more
correct to consider limiting cases for dimensionless quantities \cite{29}. Let us restate
our results in terms of relations between the time scales $T_{cl}$, $T_{coll}$, $T_{rev}$ and time
parameter $t$. Using formula (\ref{15}), we can rewrite the limits $\hbar\to 0$, $\alpha\to 0$, $\frac\hbar\alpha\to 0$ in the form
$$\frac{T_{coll}}{T_{cl}}\to\infty,\quad\frac{T_{rev}}{T_{coll}}\to\infty,\quad \frac{T_{rev}}{T_{cl}}\to\infty,$$ or
$$T_{rev}\gg T_{coll}\gg T_{cl},$$
that is, the time scales (the classical period of motion, characteristic time of collapse,
and full revival time) are distant from each other on the time line. Moreover,
the limits $\frac{\hbar t}{\alpha}(t-cT_{rev})\to2mD$ and
$\hbar t(t-cT_{rev})\to0$ occurring in Theorem~\ref{PropDynAsymp} can be rewritten in the form
$$\frac{t-cT_{rev}}{T_{coll}}\to2mD,\quad\frac t{T_{rev}}\to c.$$
Thus, we see that the limits under consideration have a clear physical meaning.
\end{remark}

\section{Coherent states in the infinite well}\label{s3}
\subsection{A map of the dynamics of a particle moving on a circle to the dynamics of a particle moving in a box (classical mechanics)}\label{s31} Consider a particle moving freely in a one-dimensional infinite potential well $[-l, l]$ with rigid (elastic) walls. In classical mechanics, this case can be reduced to dynamics on a circle. To do this, one uses a two-sheeted covering of an interval by a circle (see, for example, \cite{Kozlov3,11}). If we are given an interval $[-l,l]$ with phase variables $(q, p)$ and a circle $[-2l, 2l]$ (the points $-2l$ and $2l$ being identified) with phase variables $(q', p')$, then the two-sheeted covering of the interval by the circle is determined by the formula
\begin{equation}\label{28} 
q'=\begin{cases}q-l,&p\geq0,\\
l-q,&p<0,\end{cases}\qquad p'=|p|.
\end{equation}
Here we first shift the interval $[-l,l]$ to $[-2l,0]$, then reflect it in the point 0 and obtain the interval $[-2l, 2l]$, and then glue the points $-2l$ and $2l$. Thus, instead of oscillations on the interval, we get a rotation in one direction on a circle. We have reduced the dynamics of a particle in a box to that of a particle on a circle. However, we note that the sign of the momentum is not uniquely determined at $q = \pm l$ since it jumps at these points. This is a consequence of the approximate nature of the model.

To define the spaces of test functions and distributions on the strip $\Omega$, we first introduce the space of rapidly decaying functions on $\Omega$:
\begin{equation}\label{29} 
\begin{split}
\mathscr S(\Omega)=\{\sigma:\mathbb R^2\to\mathbb R|\quad&1)\: \sigma[(-1)^n(q+2nl),(-1)^np]=\sigma(q,p),n=0,\pm1,\pm2,\ldots;\\
&2)\: \sigma\in C^\infty(\mathbb R^2);\\
&3)\: \lim_{p\to\pm\infty}p^r\frac{\partial^{s_1+s_2}\sigma}{\partial
q^{s_1}\partial q^{s_2}}=0,\: r,s_1,s_2=0,1,2,\ldots\}.
\end{split}\end{equation}
The topology on this space is given by seminorms (\ref{17}). The space $\mathscr S'(\Omega)$ of distributions is the space of continuous linear functionals on  $\mathscr S(\Omega)$. 

Let $\sigma\in\mathscr S(\Omega)$. We define the following distributions:
\begin{align*}
&(\delta(q-q_0,p-p_0),\sigma)=\sigma(q_0,p_0),\\
&(f(q)\delta(p-p_0),\sigma)=\int_{-\infty}^{+\infty}f(q)\sigma(q,p_0),\\
&(c\delta(p-p_0),\sigma)=c\int_{-l}^l\sigma(q,p_0)\,dq,
\end{align*}
where $(q_0,p_0)\in\mathbb R^2$, $f(q)$ is an integrable function on a line, and $c\in\mathbb R$.

We again consider the function
\begin{equation*}
\varphi_D(q)=\frac1{\sqrt{2\pi
D^2}}e^{-\frac{q^2}{2D^2}},\end{equation*} where $D\in(0,\infty)$, and define the distribution $\varphi_D(q)\delta(p-p_0)$,
$p_0\in\mathbb R$, at $D=0$ and $D=\infty$ using the corresponding limits in the new space of distributions:
$$\varphi_0(q)\delta(p-p_0)=\lim\limits_{D\to0}\varphi_D(q)=\delta(q,p-p_0),$$
$$\varphi_\infty(q)=\lim\limits_{D\to\infty}\varphi_D(q)=\frac1{4l}\delta(p-p_0)+
\frac1{4l}\delta(p+p_0)$$ (the limits are taken in $\mathscr S'(\Omega)$).

Functions in $\mathscr S(\Omega)$ have the following important property:
\begin{equation}\label{30} 
\sigma(\pm l,p)=\sigma(\pm l,-p).
\end{equation}

Let $K_{2l}$ be the cylinder $K$ (defined in Subsection~\ref{s24}) with $l$ replaced by $2l$. In other words, $K_{2l}$ is the set $\Omega_{2l} = [-2l, 2l]\times\mathbb R$ with the points $(2l, p)$ and $(-2l, p)$ identified for all $p \in\mathbb R$. The space $\mathscr S(K_{2l})$ is just the space $\mathscr S(K)$ (defined in (\ref{16})) with $l$ replaced by $2l$. We easily see that $\mathscr S(\Omega)$ is a subset of $\mathscr S(K_{2l})$. Namely, $\mathscr S(\Omega) = T\mathscr S(K_{2l})$, where the map $T$ is defined by the formula

\begin{equation}\label{31} 
\sigma(q,p)=T[\sigma](q,p)=\rho(q-l,p)+\rho(l-q,-p).
\end{equation}
We define
\begin{equation}\label{32} 
\sigma(q',p')=T^{-1}[\sigma](q',p')=\begin{cases}
\frac12\sigma(l+q',p'),&q'\leq0,\\
\frac12\sigma(l-q',-p'),&q'>0.\end{cases}
\end{equation}
Here $\rho\in\mathscr S(K_{2l})$, $\sigma\in\mathscr S(\Omega)$. The map $T^{-1}$ is not an inverse: we have $TT^{-1}\sigma = \sigma$ for all functions $\sigma$ on a circle but $T^{-1}T\rho = \rho$ only for even functions, that is, those with $\rho(q', p') = \rho(-q', -p')$. 

We also introduce the space of test functions $\mathscr D(\Omega)$ by the same formula (\ref{29}) as for $\mathscr S (\Omega)$ but with the third condition of rapid decay with respect to $p$ replaced by the condition of being compactly supported with respect to $p$. We similarly introduce the space of distributions $\mathscr D'(\Omega)$. We also have $\mathscr D(\Omega) = T\mathscr D(K_{2l})$.

\subsection{A map of the dynamics of a particle moving on a circle to the dynamics of a particle moving in a box (quantum mechanics)} 
We note that finding the position (on a circle) of an image under the map (\ref{28}) requires the simultaneous knowledge of the position $q$ and momentum $p$ (or rather, the direction of the momentum) in a box, which are, thus, presupposed to be simultaneously well defined. But this is known not to be the case in quantum theory by the uncertainty relations. Therefore, one cannot use the map (\ref{28}) directly to construct the corresponding map in quantum mechanics.

We define another map
$$\Theta:L_2(-2l,2l)\to L_2(-l,l),\quad\psi(x)\mapsto[\Theta\psi](y)=\frac{\sqrt2}2[\psi(y-l)-\psi(l-y)]$$
and a map
$$\Theta^{-1}:L_2(-l,l)\to L_2(-2l,2l),\quad\varphi(y)\mapsto[\Theta^{-1}\varphi](x)=\begin{cases}
\frac{\sqrt2}2\varphi(x+l),&x\leq0,\\-\frac{\sqrt2}2\varphi(l-x),&x>0.
\end{cases}$$

The map $\Theta$ is an analogue of $T$ and $\Theta^{-1}$ is an analogue of $T^{-1}$ (see (\ref{31}) and (\ref{32})). Clearly, we have $\Theta\Theta^{-1}\varphi = \varphi$ for all functions $\varphi \in L_2(-l,l)$ but $\Theta^{-1}\Theta\psi = \psi$ only for odd functions $\psi \in L_2(-2l,2l)$. Restricting $\Theta$ to the set of all odd functions in $L_2(-2l,2l)$, we get a one-to-one correspondence between odd functions in $L_2(-2l,2l)$ and all functions in $L_2(-l,l)$, which is just the usual odd extension of functions to the interval of twice the original length. However, we define $\Theta$ for all functions from $L_2(-2l,2l)$.

\begin{proposition} 1) The map $\Theta^{-1}$ preserves scalar products:
$$(\Theta^{-1}\varphi,\Theta^{-1}\varkappa)_{2l}=(\varphi,\varkappa)_l$$
for all $\varphi,\varkappa\in L_2(-l,l)$. Here the subscripts $2l$ and $l$ indicate the scalar products in the Hilbert spaces $L_2(-2l, 2l)$ and $L_2(-l, l)$, correspondingly.

2) The map $\Theta$ preserves scalar products:
$$(\Theta\psi,\Theta\chi)_l=(\psi,\chi)_{2l},$$
if at least one of the functions $\psi\in L_2(-2l,2l)$ or $\chi\in
L_2(-2l,2l)$ is odd.
\end{proposition}
\begin{proof}
1) We have
\begin{equation*}\begin{split}
(\Theta^{-1}\varphi,\Theta^{-1}\varkappa)_{2l}&=\int_{-2l}^{2l}
\overline{[\Theta^{-1}\varphi]}(x)[\Theta^{-1}\varkappa](x)\,dx\\&=
\frac12\int_{-2l}^0\overline\varphi(x+l)\varkappa(x+l)\,dx+
\frac12\int_0^{2l}\overline\varphi(l-x)\varkappa(l-x)\,dx\\&=\int_{-l}^l
\overline\varphi(y)\varkappa(y)dy=(\varphi,\varkappa)_l;
\end{split}\end{equation*}
2) For definiteness, assume that
$\psi$ is odd. Then
\begin{equation*}\begin{split}
(\Theta\psi,\Theta\chi)_l&=
\int_{-l}^l\overline{\Theta\psi}(y)\Theta\chi(y)dy\\&=
\frac12\int_{-l}^l[\overline\psi(y-l)-\overline\psi(l-y)][\chi(y-l)-\chi(l-y)]dy\\
&= \int_{-l}^l\overline\psi(y-l)\chi(y-l)dy+
\int_{-l}^l\overline\psi(l-y)\chi(l-y)dy\\&=
\int_{-2l}^{2l}\overline\psi(x)\chi(x)\,dx=(\psi,\chi)_{2l}.
\end{split}\end{equation*}\end{proof}

\subsection{Coherent states in a box}
Consider the coherent states $\upsilon_{qp}$ (defined in (3)) on the space $L_2(-2l, 2l)$:
$$\upsilon_{qp}(x)=\sum_{n=-\infty}^{+\infty}\eta_{qp}(x-4nl),$$
where $(q,p)\in\Omega_{2l}$. Unless otherwise stated, we always assume in this subsection that $\upsilon_{qp}$ are the functions from $L_2(-2l,2l)$ defined by (\ref{3}) with $l$ replaced by $2l$.

Then we easily see that
\begin{equation}\label{33} 
[\Theta\upsilon_{qp}](y)=\begin{cases}
\frac{\sqrt2}2\,\omega_{l+q,p}(y),&q\leq0,\\
-\frac{\sqrt2}2\,\omega_{l-q,-p}(y),&q>0\end{cases}
\end{equation}
(we have used the property $\eta_{qp}(x)=\eta_{-q,-p}(-x)$).
Here
\begin{equation}\label{34} 
\omega_{qp}(y)=\sum_{n=-\infty}^{+\infty}[\eta_{qp}(y-4nl)-
\eta_{qp}(2l-y+4nl)]=
\sum_{n=-\infty}^{+\infty}(-1)^n\eta_{qp}[(-1)^n(y-2nl)],
\end{equation}
for $(q,p)\in\Omega$. One can similarly prove that
\begin{equation}\label{35} 
[\Theta^{-1}\omega_{qp}](x)=
\frac{\sqrt2}2[\upsilon_{q-l,p}(x)-\upsilon_{l-q,-p}(x)]\end{equation}
(that is, $\Theta^{-1}\Theta\upsilon_{qp}\neq\upsilon_{qp}$ since the function $\upsilon_{qp}$ is not odd). Note that $\omega_{\pm l,0}\equiv0$.

We now prove an analogue of Proposition~\ref{TheoUnit} (which actually follows from that proposition and properties of the map $\Theta$).

\begin{proposition}\label{TheoUnitSegm}The family of functions (\ref{34}) forms a continuous resolution of unity in $L_2(-l, l)$:
$$\frac{1}{2\pi\hbar}\iint_{\Omega}P[\omega_{qp}]\,dqdp=1.$$
Equality here is understood in the weak sense: for all $\varphi,\varkappa\in L_2(-l,l)$ we have
\begin{equation}\label{36} 
\frac{1}{2\pi\hbar}\iint_{\Omega}(\varphi,P[\omega_{qp}]\varkappa)\,dqdp=
\frac{1}{2\pi\hbar}\iint_{\Omega}
(\varphi,\omega_{qp})(\omega_{qp},\varkappa)\,dqdp=(\varphi,\varkappa).
\end{equation}
\end{proposition}
\begin{proof}
Using the properties of $\Theta$, $\Theta^{-1}$ and Proposition~\ref{TheoUnit}, we have
\begin{multline*}
(\varphi,\varkappa)=(\Theta^{-1}\varphi,\Theta^{-1}\varkappa)=
\frac1{2\pi\hbar}\iint_{\Omega_{2l}}
(\Theta^{-1}\varphi,\upsilon_{qp})(\upsilon_{qp},\Theta^{-1}\varkappa)\,dqdp=
\\=
\frac1{2\pi\hbar}\iint_{\Omega_{2l}} (\varphi,\Theta\upsilon_{qp})
(\Theta\upsilon_{qp},\varkappa)\,dqdp\\=
\frac1{4\pi\hbar}\iint_{\Omega_{2l}^-}
(\varphi,\omega_{q+l,p})(\omega_{q+l,p},\varkappa)\,dqdp+
\frac1{4\pi\hbar}\iint_{\Omega_{2l}^+}
(\varphi,\omega_{l-q,-p})(\omega_{l-q,-p},\varkappa)\,dqdp\\=
\frac1{2\pi\hbar}\iint_{\Omega}
(\varphi,\omega_{qp})(\omega_{qp},\varkappa)\,dqdp,
\end{multline*}
as required. Here we have written
$\Omega_{2l}^-=[-2l,0]\times\mathbb R$ and
$\Omega_{2l}^+=[0,2l]\times\mathbb R$.
\end{proof}

Proposition~\ref{TheoUnitSegm} along with some properties concerning semiclassical dynamics (to be proved below; see Theorem~\ref{PropDynAsympBox}) enables us to call the functions $\omega_{qp}$, $(q,p)\in\Omega$, coherent states in the infinite potential well (in a box).

One can obtain an analogue of formula (\ref{5}) for the functions $\omega_{qp}$:
\begin{equation*}\begin{split}\omega_{qp}(x)&=
\frac{1}{\sqrt[4]{2\pi\alpha^2}}\,
\theta\left[\frac{(x-q)l}{i\pi\alpha^2}-\frac{2pl}{\pi\hbar},
\frac{4l^2}{\pi\alpha^2}\right]\exp\left\{-\frac{(x-q)^2}{4\alpha^2}+
\frac{ip(x-q)}{\hbar}\right\}\\&- \frac{1}{\sqrt[4]{2\pi\alpha^2}}\,
\theta\left[\frac{(x-2l+q)l}{i\pi\alpha^2}-\frac{2pl}{\pi\hbar},
\frac{4l^2}{\pi\alpha^2}\right]\exp\left\{-\frac{(x-2l+q)^2}{4\alpha^2}-
\frac{ip(x-2l+q)}{\hbar}\right\}.
\end{split}\end{equation*}

Using the modular property of the theta-function, we get the following analogue
of formula (\ref{6}):
\begin{multline}\label{37} 
\omega_{qp}=\sqrt[4]{\frac{\pi\alpha^2}{32l^4}}\,\left[\theta\left(
-\frac{x-q}{4l}-\frac{p\alpha^2}{2il\hbar},\frac{\pi\alpha^2}{4l^2}\right)-
\theta\left(
-\frac{x-2l+q}{4l}-\frac{p\alpha^2}{2il\hbar},\frac{\pi\alpha^2}{4l^2}\right)
\right]e^{-(\frac{\alpha p}{\hbar})^2}=\\
\sqrt[4]{\frac{\pi\alpha^2}{2l^4}}\sum_{k=-\infty}^{+\infty}
\left[\exp\left\{-\alpha^2\left(\frac{\pi}{2l}k-\frac{p}{\hbar}\right)^2-\frac{i\pi
k(q-l)}{2l}\right\}\right.\\\left.-
\exp\left\{-\alpha^2\left(\frac{\pi}{2l}k+\frac{p}{\hbar}\right)^2+\frac{i\pi
k(q-l)}{2l}\right\}\right] \sin\left(\frac{\pi k}{2l}(x-l)\right).
\end{multline}

This yields an expansion of $\omega_{qp}$ with respect to the orthogonal basis
$$f_k=\frac{1}{\sqrt{l}}\sin\left(\frac{\pi k}{2l}(x-l)\right),\quad k=1,2,\ldots.$$
As in the circle case, the same result can be obtained by a direct calculation of the scalar product of $\omega_{qp}$ with the elements of this basis if we use the formula
$$\int_{-l}^l\omega_{qp}(x)\sin\left(\frac{\pi k}{2l}(x-l)\right)dx=
\int_{-\infty}^{+\infty}\eta_{qp}(x)\sin\left(\frac{\pi k}{2l}(x-l)\right)dx$$
for integer $k$. This formula is obtained from the expression
$$\int_{-2l}^{2l}\upsilon_{qp}(x)\sin\left(\frac{\pi}{2l}kx\right)dx=
\int_{-\infty}^{+\infty}\eta_{qp}(x)\sin\left(\frac{\pi}{2l}kx\right)dx$$
(which is proved in a similar way to (\ref{7})) using properties of $\Theta$:
\begin{equation*}\begin{split}\int_{-\infty}^{+\infty}&\eta_{qp}(x)\sin\left(\frac{\pi k}{2l}(x-l)\right)dx=\int_{-\infty}^{+\infty}\eta_{q-l,p}(x)\sin\left(\frac{\pi k}{2l}x\right)dx\\&=
\int_{-2l}^{2l}\upsilon_{q-l,p}(x)\sin\left(\frac{\pi}{2l}kx\right)dx=
\left(\sin\left(\frac{\pi}{2l}kx\right),\upsilon_{q-l,p}\right)_{2l}\\&=
\left(\Theta\left[\sin\left(\frac{\pi}{2l}kx\right)\right],\Theta\upsilon_{q-l,p}\right)_l=
\left(\sin\left(\frac{\pi}{2l}k(x-l)\right),\omega_{qp}\right)_l\\&=
\int_{-l}^l\omega_{qp}(x)\sin\left(\frac{\pi k}{2l}(x-l)\right)\,dx.
\end{split}\end{equation*}

The functions $f_k$ are eigenfunctions of the operator
$$H^b=-\frac{\hbar^2}{2m}\frac{d^2}{dx^2}$$
with the domain
$$D(H^b)=\{\psi\in AC^2(-l,l)|\,\psi(-l)=\psi(l)=0\}.$$
The operator $H^b$ is a Hamiltonian (energy operator) for a free (see Remark~\ref{r1} for the use of the notion ``free'') quantum particle in the infinite square potential well (see \cite{9}). Thus, the functions $\omega_{qp}(x)$ can be expanded in uniformly convergent series with respect to the eigenfunctions of the Hamiltonian for a particle in a box.
The evolution of the state $\omega_{qp}$ over time is given by the formula
$$\omega_{qp,t}=U^b_t\omega_{qp},$$ where
\begin{equation}\label{38} 
U^b_t=\exp(-\frac{it}{\hbar}H^b)
\end{equation}
is the evolution operator for a free quantum particle in a box. The function $\omega_{qp,t}$ satisfies the Schr\"{o}dinger equation with the boundary conditions corresponding to the infinite square well:
\begin{equation}\label{39}
\begin{split}
&i\hbar\frac{\partial\omega_{qp,t}}{\partial t}=-\frac{\hbar^2}{2m}\frac{\partial^2\omega_{qp,t}}{\partial x^2},\\
&\omega_{qp,t}(-l)=\omega_{qp,t}(l)=0,\\
&\omega_{qp,0}(x)=\omega_{qp}(x),
\end{split}
\end{equation}
where $x\in[-l,l]$, $t\in\mathbb R$. Using the reflection method, we arrive at
$$\omega_{qp,t}(x)=\sum_{n=-\infty}^{+\infty}(-1)^n\eta_{qp,t}[(-1)^n(x-2nl)],$$
where $\eta_{qp,t}$ is defined as above (see (\ref{11})).

Note that formulae (\ref{33}) and  (\ref{35})
remain valid if we replace $\upsilon_{qp}$ and $\omega_{qp}$ by $\upsilon_{qp,t}$ and $\omega_{qp,t}$ (respectively) for an arbitrary $t$.

Expanding in eigenfunctions of $H^b$, we get

$$\omega_{qp,t}(x)=\sum_{k=1}^{\infty}b_{k,qp}\sin\left(\frac{\pi}{l}kx\right)\exp\left\{-\frac{i\hbar t}{2m}
\left(\frac{\pi k}{2l}\right)^2\right\},$$ where the coefficients $b_{k,qp}$ can be found from formula (\ref{37}).

The dynamics of coherent states $\omega_{qp,t}$ in a box is similar to that of coherent
states on a circle. For small values of time, a well-localized initial wave packet remains well-localized and its centre moves along the classical trajectory with period
$$T_{cl}=\frac{4lm}p.$$
The wave packet eventually collapses, and at the moment
$$T_{coll}=\frac{2ml\alpha}{\sqrt3\hbar}$$ we observe an approximately uniform spatial distribution. However, at the moment
$$T_{rev}=\frac{16ml^2}{\pi\hbar}$$
we observe a full revival of the wave packet:
$\omega_{qp,T_{rev}}=\omega_{qp}$. Fractional revivals are observed at the moments $\frac{M}NT_{rev}$ for all integer $M$ and $N$. The full revival time now increases by a factor of 4 compared to that for motion on a circle (compare the last formula for $T_{rev}$ with formula (\ref{14}))
because the energy spectrum is now twice as dense as that for a particle on the circle $[-l,l]$.

Thus, as in the case of a particle on a circle, we have three time scales:

\begin{enumerate}[1)]
\item $T_{cl}$, the classical period of motion;

\item $T_{coll}$, the characteristic time of collapse of the quantum wave packet;

\item $T_{rev}$, the period of full revival of the quantum wave packet.
\end{enumerate}
The asymptotic behaviour of these time scales as $\hbar,\alpha,\frac\hbar\alpha\to0$ is again described by formula (\ref{15}).

\subsection{The semiclassical limit of the dynamics of coherent states of a particle in a box}
We now prove an analogue of Theorem~\ref{PropDynAsymp} for the box.
\begin{theorem}\label{PropDynAsympBox} We have the following limit formula in $\mathscr S'(\Omega)$  (where the variables $(q,p)$ are fixed and $(q',p')$ are variables of integration with test functions $\sigma(q',p')\in\mathscr S(\Omega)$):
\begin{equation}\label{40} 
\lim\{\frac{1}{2\pi\hbar}|(\omega_{qp},\omega_{q'p',t})|^2
-\frac1{N'}\sum_{k=0}^{N'-1}\varphi_D[q'-q-\frac{4kl}{N'}-a+\frac
pm(t-cT_{rev})]\delta(p'-p)\}=0.
\end{equation}
The limit is performed as follows: $\hbar\to0$, $\alpha\to0$,
$\frac\hbar\alpha\to0$, $t=t(\hbar)$,
$\frac{\hbar}\alpha(t-cT_{rev})\to2mD$, $\hbar(t-\frac
cT_{rev})\to0$, where $c\in\mathbb R$, $D\in[0,\infty]$ and the numbers $N'$
and $a$ depend on $c$ and $N$. If $c$ is rational (and so, can be written as a reduced fraction $c=\frac MN$), then
$N'=N$ for odd $N$ and $N'=\frac N2$ for even $N$. Further,
$a=\frac{2l}N$ for $N\equiv 2\pmod 4$ and $a=0$ otherwise. If
 $c$ is irrational, then $N'=1$, $a=0$. The parameter $\alpha$
occurs in the definition of $\omega_{qp}$ (see
(\ref{1}) and (\ref{34})), $T_{rev}=\frac{16ml^2}{\pi\hbar}$.

The convergence in (\ref{40}) is uniform with respect to $(q,p)$ on every subset of
$\Omega$ disjoint from some neighbourhood of the closed interval
$\{p=0\}\subset\Omega$. Moreover, if $c$ is integer or half-integer, then the convergence is uniform on every subset of $\Omega$ disjoint from some neighbourhoods of the points $(\pm l,0)$.
\end{theorem}

Since the meaning of the assertions of Theorem~\ref{PropDynAsympBox} is analogous to that in Theorem~\ref{PropDynAsymp}, we repeat it only briefly here. Some complications appear, in particular,
because we must take into account that the particle reflects from the walls and,
hence, the sign of its momentum changes. Here $\frac1{2\pi\hbar}|(\omega_{qp},\omega_{q'p',t})|^2$
is the probability density for a quantum particle to be in the state $\omega_{qp}$ at time $t$ under the condition that it was in the state $\omega_{q'p'}$ at time 0 (up to the norms $\|\omega_{qp}\|^2$ and
$\|\omega_{q'p'}\|^2$, which tend to 1 in the sense of our limit on the sets under consideration; see Proposition~\ref{LemNormBox} below.

In the semiclassical limit $\hbar,\alpha,\frac\hbar\alpha\to0$ a quantum particle in the state
$\omega_{qp}$ has a well-defined position (equal to $q$) and momentum
(equal to $p$),just like a classical particle. Therefore, we can say that, in the semiclassical limit,
$\frac1{2\pi\hbar}|(\omega_{qp},\omega_{q'p',t})|^2$ is the probability density for a quantum particle in a box to be at the phase point $(q,p)$ at time $t$ under the condition that it was at the phase point $(q',p')$ at time 0.

As above, the case of $c=0$ and $D=0$ ($\frac{\hbar t}\alpha\to0$)
corresponds to the classical time scale $T_{cl}$:
time is either fixed or increases slower than the rate of decrease of the collapse velocity (proportional to $\frac\hbar\alpha$).of the packet. Then the limit formula (\ref{40}) takes the form
\begin{equation*}
\lim[\frac{1}{2\pi\hbar}|(\omega_{qp},\omega_{q'p',t})|^2-
\delta(q'-q+\frac{p}mt,p'-p)]=0.
\end{equation*}
Here $\delta(q' - q + \frac{pt}m, p' - p)$ is the probability density for a classical particle in a box to be at the phase point $(q, p)$ at time $t$ under the condition that it was at the phase point $(q',p')$ at time 0.  Thus, in the semiclassical limit at time scale $T_{cl}$, we have classical dynamics: the quantum probability density of transition to $(q,p)$ for a particle that was at $(q',p')$ at time 0 is equal to the corresponding classical probability density.

The case $c=0$, $D\in(0,\infty)$ ($\frac{\hbar
t}\alpha\to2mD\in(0,\infty)$) corresponds to the second time scale
$T_{coll}$. In this case we observe some spatial spread of the probability distribution. In contrast to the circle case, there is flattening not only of the spatial probability density but also of the probability of the signs of the momentum.

The case $c=0$, $D=\infty$ ($\frac{\hbar t}\alpha\to\infty$)
corresponds to a complete flattening of the spatial probability density (and the probability of the signs of the momentum): formula (\ref{40}) takes the form
\begin{equation}\label{41} 
\lim\frac{1}{2\pi\hbar}|(\omega_{qp},\omega_{q'p',t})|^2=
\frac1{4l}[\delta(p'-p)+\delta(p'+p)].
\end{equation}
Therefore, we also associate this case with the time scale $T_{coll}$ (corresponding to destruction of the localized wave packet). This yields a mathematical justification of the (asymptotic) flattening of the spatial probability density for a quantum particle in a box, which was previously known only from numerical experiments \cite{5,6}.

The case $c\neq0$ corresponds to the third time scale $T_{rev}$.
If $c$ is irrational, then, as in
the previous case, formula (\ref{40}) reduces to (\ref{41}), that is, one observes a complete
flattening of the spatial probability density. The case of rational $c$ corresponds to
a revival of the wave packet (full if $N'=1$, that is, $c$ is integer or half-integer, and fractional otherwise). The case $D>0$ corresponds to the spread revived wave packets. In case $D=\infty$, we again obtain a complete flattening of the spatial density of the distribution 
(\ref{41}).

The difference from the circle case regarding the condition $p\neq 0$ for $N\neq1,2$ (see the statement of Theorem~\ref{PropDynAsympBox}) is explained by a specific structure of fractional revivals of the states $\omega_{q0}$
coming from the symmetry of these states (namely, from the property
$\omega_{2l-q,0}=-\omega_{q0}$).

As in the circle case, since every irrational number can be approximated by rationals, we can (heuristically) say that the case of irrational $c$ in Theorem~\ref{PropDynAsympBox} is a limiting case of rational $c$ as $N\to\infty$: the distance between neighbouring terms in the sum of delta-functions in (\ref{40}) tends to zero, and the sum of the delta-functions tends to the uniform distribution (in the weak sense). In other words, the cases of an irrational $c$ and a very close rational $c'$ are almost indistinguishable.

Note that the distribution of the modulus of the momentum is preserved in all these limiting cases.

Thus, we have traced the whole dynamics of a quantum wave packet in the infinite potential well. A well-localized wave packet eventually collapses with a complete flattening of the spatial density and the momentum sign. At some moments we see that copies of the initial packet simultaneously arise at several points of the well and then again eventually collapse with a complete flattening of the density. Since there are ``more'' irrationals than rationals, we can say that the particle most often stays in states whose spatial density of distribution is close to the uniform density. Thus, Theorem~\ref{PropDynAsympBox} completely describes the free quantum dynamics of a particle in the infinite potential well at all time scales the semiclassical limit. All stages are parametrized by the two real parameters $c$ and $D$.

Let us formulate a simplified version of Theorem~\ref{PropDynAsympBox}, which deals only with principal time scales (classical motion, complete flattening, and exact revivals) without intermediate ones.

\begin{corollary}
We have the following limit formulae in $\mathscr S'(\Omega)$:

1) \begin{equation*}
\lim[\frac{1}{2\pi\hbar}|(\omega_{qp},\omega_{q'p',t})|^2-
\delta(q'-q+\frac{p}mt,p'-p)]=0
\end{equation*}
as $\hbar,\alpha,\frac\hbar\alpha\to0$, $t=const$;

2) $$\lim\frac{1}{2\pi\hbar}|(\omega_{qp},\omega_{q'p',t})|^2=
\frac1{4l}\delta(p'-p)+\frac1{4l}\delta(p'+p)$$
as $\hbar,\alpha,\frac\hbar\alpha\to0$, $t\to\infty$, $\frac{\hbar t}{\alpha}\to\infty$,\\
as well as $\hbar,\alpha,\frac\hbar\alpha\to0$, $t=cT_{rev}\to\infty$, where $c$
is irrational, $T_{rev}=\frac{16ml^2}{\pi\hbar}$;

3)
\begin{equation*}
\lim[\frac{1}{2\pi\hbar}|(\omega_{qp},\omega_{q'p',t})|^2
-\frac1{N'}\sum_{k=0}^{N'-1}\delta(q'-q-\frac{4kl}{N'}-a,p'-p)]=0
\end{equation*}
as $\hbar,\alpha,\frac\hbar\alpha\to0$, $t=cT_{rev}\to\infty$,
where $c=\frac MN$ is rational and the numbers $N'$ and $a$ depend on $N$.

The limits in cases 1), 2) and 3) for $N=1$ or $N=2$ (that is, for integer or half-integer $c$) are uniform with respect to $(q,p)$ on every subset of $\Omega$ disjoint from some neighbourhoods of the points  $(\pm l,0)$. The convergence in case 3) for $N\neq1$, $N\neq2$ is uniform on every subset of $\Omega$, disjoint from some neighbourhood of the closed interval $\{p=0\}\subset\Omega$.
\end{corollary}

\begin{proof}[Proof of Theorem~\ref{PropDynAsympBox}]
Since 
$$\omega_{q'p'}=\sqrt2\Theta\upsilon_{q'-l,p'}=
\frac{\sqrt2}2\Theta[\upsilon_{q'-l,p'}-\upsilon_{l-q',-p'}]$$
and $(\upsilon_{q'-l,p'}-\upsilon_{l-q',-p'})$ is an odd function, we obtain from the properties of $\Theta$ that
$$(\omega_{qp},\omega_{q'p',t})_l=
(\upsilon_{q-l,p},\upsilon_{q'-l,p',t}-\upsilon_{l-q',-p'})_{2l}$$ (we will omit the subindex $2l$ in the remain part of the proof). We have
\begin{multline*}
\iint_\Omega\frac1{2\pi\hbar}|(\upsilon_{q-l,p},
\upsilon_{q'-l,p',t}-\upsilon_{l-q',-p',t})|^2\sigma(q',p')\,dq'dp'\\=
\iint_{\Omega_{2l}}\frac1{2\pi\hbar}|(\upsilon_{q-l,p},
\upsilon_{q'p',t}-\upsilon_{-q',-p',t})|^2T^{-1}[\sigma](q',p')\,dq'dp'.
\end{multline*}
Furthermore,
\begin{multline}\label{42} 
\frac1{2\pi\hbar}|(\upsilon_{q-l,p},
\upsilon_{q'p',t}-\upsilon_{-q',-p',t})|^2=
\frac1{2\pi\hbar}|(\upsilon_{q-l,p},\upsilon_{q'p',t})|^2+
\frac1{2\pi\hbar}|(\upsilon_{q-l,p},\upsilon_{-q',-p',t})|^2\\-
\frac1{\pi\hbar}\Re(\upsilon_{q-l,p},\upsilon_{q',p',t})
(\upsilon_{-q',-p',t},\upsilon_{q-l,p}).
\end{multline}
Clearly, the last summand tends to zero uniformly on every subset of $\Omega$, disjoint from some neighbourhood of the closed interval $\{p=0\}$ (that is, one can find $\varepsilon$ such that $|p|>\varepsilon$ for all elements of the subset). 

Consider a weaker restriction: $(q,p)$ lies in a subset of $\Omega$, disjoint from some neighbourhoods of the points $(\pm l,0)$. Since we already know that the summand tends to zero for $p\neq0$, we consider the case $p=0$. Let us expand the function $T^{-1}[\sigma]$ into a Fourier series with respect to $q'$ and represent it by a Fourier integral with respect to $p'$:
$$T^{-1}[\sigma](q',p')=\frac1{2\sqrt{\pi l}}\sum_{j=-\infty}^{+\infty}\int_{-\infty}^{+\infty}a_j(\nu)e^{i\frac\pi{2l}jq'+i\nu p'}\,dq'dp'.$$
Using formula (\ref{21})  and the method of the proof of Theorem~\ref{PropDynAsymp}, we arrive at the formula
\begin{multline*}
\frac1{2\pi\hbar}\iint_\Omega (\upsilon_{q-l,0},\upsilon_{q'p',t})
(\upsilon_{-q',-p',t},\upsilon_{q-l,0})e^{i\frac\pi{2l}jq'+i\nu
p'}\,dq'dp'\\=\sum_{r=-\infty}^{+\infty} \exp\left\lbrace
-\frac{\alpha^2(4+\gamma^2)}4\left[ \frac{\pi
j}{2l}-\frac{\gamma(q-l-2rl)}{\alpha^2(4+\gamma^2)}\right]^2
-\right.\\\left. -\frac{1}{4\alpha^2}\left[q-l-2rl-\frac{\pi j\hbar
t}{4ml}+\nu\hbar\right]^2-\frac{(q-l-2rl)^2}{\alpha^2(4+\gamma^2)}\right\rbrace.
\end{multline*}
This expression tends to zero for $q\neq\pm l$ in the limiting cases corresponding to the case of rational $c$ with $N=1$ or $N=2$ and, under this condition, the convergence is uniform on every interval $[-l+\varepsilon,l-\varepsilon]$, $\varepsilon>0$. That is,
\begin{equation}\label{43} 
\lim\frac1{2\pi\hbar}\iint_\Omega
(\upsilon_{q-l,p},\upsilon_{q',p',t})
(\upsilon_{-q',-p',t},\upsilon_{q-l,p})T^{-1}[\sigma](q',p')\,dq'dp'=0
\end{equation}
uniformly on $(q,p)$ on every subset of $\Omega$ disjoint from some neighbourhoods of the points $(\pm l,0)$.

Using equation (\ref{42}) with zero last term, Theorem~\ref{PropDynAsymp}, the definition of $T^{-1}$ and the eveness property of $T^{-1}[\sigma](q',p')=T^{-1}[\sigma](-q',-p')$
of the image of $T^{-1}$, we obtain for rational $c$ that
\begin{equation*}\begin{split}
0&=\lim\iint_\Omega\frac1{2\pi\hbar}|(\omega_{qp},\omega_{q'p',t})|^2
\sigma(q',p')\,dq'dp'\\&- \frac2{N'}\sum_{k=0}^{N'-1}\int_{-2l}^{2l}
\varphi_D[q'-q+l-\frac{4kl}{N'}-a+\frac pm(t-\frac
MNT_{rev})]T^{-1}[\sigma](q',p)\,dq'
\\
&=\lim\iint_\Omega\frac1{2\pi\hbar}|(\omega_{qp},\omega_{q'p',t})|^2
\sigma(q',p')\,dq'dp'
\\
&-\frac1{N'}\sum_{k=0}^{N'-1}\int_{-l}^{l}
\varphi_D[q'-q-\frac{4kl}{N'}-a+\frac pm(t-\frac
MNT_{rev})]\sigma(q',p)\,dq'\\
&-\frac1{N'}\sum_{k=0}^{N'-1}\int_{-l}^{l}
\varphi_D[q'-2l+q+\frac{4kl}{N'}+a-\frac pm(t-\frac
MNT_{rev})]\sigma(q',-p)\,dq',
\end{split}\end{equation*}
where $N'$ and $a$ are defined as in the statement of the theorem. This proves the theorem in the case of rational $c$.

For an irrational $c$ we have
\begin{multline*}
\lim\iint_\Omega\frac1{2\pi\hbar}|(\omega_{qp},\omega_{q'p',t})|^2
\sigma(q',p')\,dq'dp'\\=\frac1{4l}\int_{-2l}^{2l}
\{T^{-1}[\sigma](q',p)+T^{-1}[\sigma](q',-p)\}\,dq'=\frac1{4l}\int_{-2l}^{2l}
\sigma(q',p)\,dq'.
\end{multline*}
The uniform convergence in the two last formulae follows from that in (\ref{19}) and the uniform convergence to zero of the expression (\ref{43}) or the last summand in (\ref{42}) (depending on the limiting case under consideration). 

The theorem is proved.
\end{proof}

Now let us prove an analogue of Proposition~\ref{LemNormCirc}.

\begin{proposition}\label{LemNormBox}
The norm of $\omega_{qp}$ tends to unity as $\hbar\to0$,
$\alpha\to0$, $\frac\hbar\alpha\to0$. The convergence is uniform on every subset of $\Omega$ disjoint from some neighbourhoods of the points $(\pm l,0)$.
\end{proposition}
\begin{proof}
Using the properties of $\omega_{qp}$ and $\Theta$, we have the following chain of equalities (see the beginning of the proof of Theorem~\ref{PropDynAsympBox}):
$$\|\omega_{qp}\|^2=(\omega_{qp},\omega_{qp})_l=
(\upsilon_{q-l,p},\upsilon_{q-l,p}-\upsilon_{l-q,-p})_{2l}=
\|\upsilon_{q-l,p}\|^2-(\upsilon_{q-l,p},\upsilon_{l-q,-p})_{2l}.$$
The first summand tends to 1 by Proposition~\ref{LemNormCirc}. The second summand tends to zero for
$(q,p)\in\Omega$, $(q,p)\neq(\pm l,0)$. Both limits are uniform on the sets indicated in the statement of the proposition. The proposition is proved.
\end{proof}

We make three remarks on Theorem~\ref{PropDynAsympBox}.

\begin{remark}\label{r5}The remarks at the end of Section~\ref{s2}, which were made in the circle case, remain valid here except for some inessential differences in the semiclassical dynamics of the mean position and mean momentum. First, the mean momentum is not preserved in the semiclassical limit in a box but periodically changes the sign (because of reflections in the walls) while its absolute value is preserved. Second, since the position of a particle in a box is well defined, we can use a familiar formula for the mean position:
$$\overline q_t=\int_{-l}^lx|\omega_{qp}(x)|^2\,dx.$$
Moreover, the centre of the packet is easily seen to move along the classical trajectory at the time scale $T_{cl}$.
\end{remark}

\begin{remark}\label{r6}We have proved theorems on the semiclassical limit of quantum dynamics on a circle and on an interval (in a box). Of course, these results can easily be extended to the case of multidimensional domains that are Cartesian products of any number of circles and intervals: rectangles, rectangular parallelepipeds, tori, cylinders, and so on. It would be interesting to extend our results to the case of more general domains, including arbitrary domains in three-dimensional space, arbitrary compact manifolds, as well as to the case of interacting particles.
\end{remark}

\begin{remark}
It is worthwhile to mention here the interesting results of M.\,V.~Berry about the fractal images of the graph of $|\psi(x,t)|^2$ as a function of $x$ and $t$. Here $\psi(x,t)$ is a solution of the Schr\"odinger equation for the infinite well (\ref{39}) with the uniform initial state $\psi(x,0)=\frac1{\sqrt{2l}}$. In fact, the result is valid for a box of arbitrary shape and finite surface area in any finite dimensional space \cite{Berry1}.

Also, surprisingly, quantum revivals have an analogue in classical optics (Talbot interference) \cite{Berry2,Berry3}.
\end{remark}

\section{The semiclassical limit of the dynamics for Husimi functions}\label{s4}
\subsection{The case of free quantum dynamics on a circle}
A new (so-called \textit{functional}) formulation of classical mechanics (or \textit{functional mechanics}) was suggested in \cite{14,15} (see also \cite{16,Mikh,30,Pisk,17,18}). The basic concept of functional mechanics is not a material point or an individual trajectory, but the probability density function in a phase space. Accordingly, the fundamental dynamical equations are not the Newton (or, equivalently, Hamilton) equations but the Liouville equation (even in the case of one particle, not an ensemble). The Newton (Hamilton) equations become approximate equations for the mean values of the positions and momenta. Corrections to solutions of the Newton equations have been calculated in some particular cases \cite{14,15,30,Pisk}.

The aim of passing to the new formulation of classical mechanics is to achieve the compatibility of the reversible microscopic dynamics and the irreversible macroscopic dynamics. This problem is known as the irreversibility problem (or reversibility paradox). It is one of the most fundamental problems of mathematical physics.

The dynamics of a material point is known to be reversible and recurrent while the dynamics of the density function satisfies the Liouville equation and has the so-called delocalization (collapse) property, which corresponds to irreversible behaviour. Therefore, if we adopt the description in terms of density functions not only for many-particle systems but even for one particle, then there is no contradiction between the micro- and macroscopic dynamics: both are irreversible (in some sense).

A procedure of constructing the density function of a physical system from directly observable quantities (results of measurements) is described in \cite{17}. The interaction of a system and the measuring instrument is studied from the point of view of functional mechanics in \cite{18}.

Now we want to know whether functional mechanics is preferable to Newtonian one from the quantum-mechanical point of view. To do this, we consider the semiclassical limits for quantum dynamics on a circle and in a box. In this subsection we take the case of the circle.

There is a correspondence between quantum states (density operators) and classical states (distribution density functions), which was considered by Husimi in the case of coherent states on the whole axis \cite{31,32,33}. Here we define an analogous correspondence for coherent states on a circle.

Let $\rho$ be a density operator (in other words, a quantum state) in $L_2(-l,l)$, that is, a positive operator with unit trace. We associate with $\rho$ the following function on the phase space $\Omega$:

\begin{equation}\label{44} 
\widetilde\rho(q,p)=\frac{1}{2\pi\hbar}\Tr P[\upsilon_{qp}]\rho=
\frac{1}{2\pi\hbar}(\upsilon_{qp},\rho\upsilon_{qp}),
\end{equation}
where $\Tr$ denotes the trace of an operator. If the density operator is a projector ($\rho = P[\psi]$, $\psi \in L_2(-l,l)$), which corresponds to the case of a pure quantum state, then formula (\ref{44}) takes the form 
\begin{equation*}
\widetilde\rho(q,p)=\frac{1}{2\pi\hbar}|(\psi,\upsilon_{qp})|^2.
\end{equation*}
Clearly, $\widetilde\rho(q, p) > 0$ and, by (\ref{4}), we have
$$\iint_\Omega\widetilde\rho(q,p)\,dqdp=1,$$
whence $\widetilde\rho(q,p)$ is a probability density function on the phase space (that is, a classical state). The correspondence (\ref{44}) taking each quantum density operator to a classical density of probability distribution is called the \textit{Husimi transform}. It can also be expressed as smoothing the Wigner function of the quantum state with a Gaussian function \cite{31,32,33}. The probability density function $\widetilde\rho$ is called the \textit{Husimi function} of the operator $\rho$.

\begin{remark}The Husimi function is not the only way of mapping  quantum density operators to  classical probability density functions. One problem of the Husimi function is that its marginal distributions of position and of momentum does not coincide with the corresponding quantum-mechanical distributions. For example, if (for simplicity) $\rho=P[\psi]$ for some $\psi\in L_2(-l,l)$, then the equality
\begin{equation}\label{EqMarg}
\int_{-\infty}^{+\infty}\widetilde\rho(x,p)\,dp=|\psi(x)|^2,
\end{equation}
in general, does not hold. Another way of mapping quantum density operators to  classical probability density functions is the so called tomography map \cite{Manko1,Manko2,Manko3}. A relation of the tomography map to the Husimi function is discovered in \cite{Manko4}.

However, we take the Husimi transform for the following reason. Consider the positive operator-valued measure $M$ given by the formula
\begin{equation}\label{EqPOVM}
M(B)=\frac1{2\pi\hbar}\iint_BP[\upsilon_{qp}]\,dqdp,
\end{equation}
where $B\subset\Omega$ is a Borel set. It can be regarded as an approximate simultaneous measurement of the position and momentum of a quantum particle on a circle (such measurements were introduced by J.~von Neumann \cite{Neumann}).

We would like to analyse the correspondence between classical and quantum dynamics on a circle. In classical mechanics, a simultaneous measurement of the position and momentum is allowed. Quantum-mechanically, this corresponds to an approximate measurement of the position and momentum like (\ref{EqPOVM}).

If a quantum particle is in the state $\rho$ and we perform measurement (\ref{EqPOVM}), then the probability density function of the result is exactly $\widetilde\rho(q,p)$ given by (\ref{44}).

Moreover, we will consider the Husimi function only in the semiclassical limit. In this limit, its marginal distributions coincide with the quantum-mechanical distributions (in particular, (\ref{EqMarg}) holds).
\end{remark}

More generally, a classical state $\sigma$ is a distribution\footnote{In the most general case, a classical state is defined as a probability measure on the phase space \cite{34}.}. Suppose that $\sigma\in\mathscr D'(K)$
(see Subsection~\ref{s31} for the definition of $\mathscr D'(K)$). The non-negativity of $\sigma$ is understood as $(\sigma,f)\geq0$ for every non-negative
test function $f$. The unit normalization of a distribution means that there is a limit $\lim\limits_{P\to\infty}
(\sigma,\lambda_P)=1$, where $\lambda_P$, $P>0$, is a family of functions such that $\lambda(x,p)=1$ for $|p|\leq P$,
$0\leq\lambda(x,p)\leq1$ for $P<|p|<P+1$, and $\lambda(x,p)=0$ for
$|p|\geq P+1$. If $\sigma(q,p)$ is an ordinary function (a regular distribution), this condition expresses its integrability on $\Omega$ and $\iint_\Omega\sigma \,dqdp=1$. We denote the set of all non-negative and unit-normalized distributions from $\mathscr D'(K)$ by $\mathscr D'_1(K)$.

We would like to know whether one can obtain any classical state as a semiclassical limit of the Husimi functions of some family of quantum states. This question
can be posed rigorously as follows. Given a (distribution) probability density function  $\sigma\in\mathscr
D'_1(K)$, can one find a family $\rho^{(\hbar)}$, $\hbar>0$,
of density operators in $L_2(-l,l)$ such that
$\widetilde\rho^{(\hbar)}\to\sigma$ in $\mathscr D'(K)$ as
$\hbar\to0$ (where $\widetilde\rho^{(\hbar)}$ is the Husimi function of the operator
$\rho^{(\hbar)}$)? We note that the definition of the functions
$\upsilon_{qp}$ occurring in the Husimi transform
involves a parameter $\alpha$. As in Theorem~\ref{PropDynAsymp}, we shall assume that $\alpha=\alpha(\hbar)$, $\alpha\to0$,
$\frac\hbar\alpha\to0$ (we have already noted that in this limiting case a quantum particle in the
state $\upsilon_{qp}$ has well-defined position and momentum).

A further question concerns the dynamics. Let the quantum system (the density
operator) evolve in time on a circle:
$\rho^{(\hbar)}_t=U^c_t\rho^{(\hbar)} U_t^{c\dag}$, where $U^c_t$
is defined by (\ref{8}). We denote the Husimi function of
$\rho^{(\hbar)}_t$ by $\widetilde\rho_t^{(\hbar)}(q,p)$.
The evolution of
the classical system (the probability density function) is given by the formula $\sigma_t(q,p)=\sigma(q-\frac{p}{m}t,p)$. We may ask whether the correspondence between the
family of quantum states $\rho^{(\hbar)}$, $\hbar>0$, and the classical state $\sigma$ is preserved in time
at various time scales (see Subsection~\ref{s23}). Namely, is it true that
$\widetilde\rho_t^{(\hbar)}-\sigma_t\to0$ in $\mathscr D'(K)$ as
$\hbar\to0$, where $t\in\mathbb R$ is an arbitrary fixed number (time scale $T_{cl}$)) or $t\to\infty$ (time scales $T_{coll}$ and $T_{rev}$; see the various versions of the relation between the limits $\hbar\to0$
and $t\to\infty$ in Theorem~\ref{PropDynAsymp})?

If the answers to these questions differ for $\sigma$ of the form
$\sigma(q,p)=\delta(q-q_0,p-p_0)$ (which corresponds to an individual trajectory of a material point and, hence, to Newtonian mechanics) and for $\sigma\in L_1(\Omega)$
 (which corresponds to a ``bunch'' of trajectories and, hence, to functional mechanics), then one can say that one formulation or the other is preferable.

\begin{remark}\label{r7}
It is known that the asymptotic behaviour of $\sigma_t$ depends strongly on
the initial density function $\sigma$. If
$\sigma(q,p)=\delta(q-q_0,p-p_0)$, then the motion is periodic and the period is $\frac{lm}{p_0}$. If $\sigma\in L_1(\Omega)$, then
the spatial density asymptotically flattens with respect to $q$ on a circle (in the sense of the so-called weak limit). Namely, the Kozlov's second  theorem on diffusion (Theorem~2 in \cite{Kozlov3}) says that there is a limit
\begin{equation}\label{45} 
\lim_{t\to\pm\infty}\iint_{\Omega}\sigma(q-\frac pmt,p)g(q,p)\,dqdp=
\int_{-\infty}^{+\infty}\left[\frac{1}{2l}\int_{-l}^l\sigma(q',p)\,dq'\right]
g(q,p)\,dqdp,\end{equation} if $\sigma,g\in L_2(\Omega)$. Formula  (\ref{45}) is easily seen to be valid for
$\sigma\in L_1(\Omega)$, $g\in\mathscr D(K)$. Thus,
\begin{equation}\label{46} 
\lim_{t\to\pm\infty}\sigma(q-\frac
pmt,p)=\frac1{2l}\int_{-l}^l\sigma(q',p)\,dq'\quad \text{in }\mathscr
D'(K),\end{equation} if $\sigma\in L_1(\Omega)$. This is the sense in which the spatial density flattens.
\end{remark}

The following theorem answers the questions posed above.

\begin{theorem}\label{TheoCoherDynClCirc}
1) Let $\sigma\in\mathscr D_1'(K)$. Then there is a family of density operators $\rho^{(\hbar)}$, $\hbar>0$, in $L_2(-l,l)$
 such that the corresponding Husimi functions
$\widetilde\rho^{(\hbar)}$converge to $\sigma$ in $\mathscr D'(K)$
as $\hbar\to0$, $\alpha\to0$, $\frac\hbar\alpha\to0$.

2) Let, further, $\rho^{(\hbar)}_t=U^c_t\rho^{(\hbar)}U_t^{c\dag}$, and let
$\widetilde\rho^{(\hbar)}_t$ be the Husimi function of $\rho^{(\hbar)}_t$. Then
\begin{equation}\label{47} 
\lim[\widetilde\rho^{(\hbar)}_t(q,p,t)-
\sigma(q-\frac{p}{m}t,p)]=0\quad\text{in } \mathscr
D'(K).\end{equation} Here the limit is performed as follows:
$\hbar\to0$, $\alpha\to0$, $\frac\hbar\alpha\to0$, $t=t(\hbar)$,
$\frac{\hbar t}\alpha\to0$.

3) If $\sigma\in L_1(\Omega)$, then equality
(\ref{47}) remains valid in the
following limit: $\hbar\to0$, $\alpha\to0$,
$\frac\hbar\alpha\to0$, $t\to\infty$, $\hbar t\to0$, $\frac{\hbar
t}\alpha\to2mD\in(0,\infty]$.
\end{theorem}

\begin{proof}
let $\sigma=\sigma(q,p)$ be a function from $\mathscr S(K)$ (a particular case of distribution from $\mathscr D'(K)$) satisfying the conditions of non-negativity and unit normalization. Then we define
\begin{equation}\label{48} 
\rho^{(\hbar)}=\iint_\Omega \sigma(q',p')P[\upsilon_{q'p'}]
\frac{dq'dp'}{\|\upsilon_{q'p'}\|^2}.\end{equation} The right-hand side depends on $\hbar$ because of the definition of the function
$\upsilon_{q'p'}$. Moreover, we choose $\alpha$ in the definition of $\upsilon_{qp}$ to be a function of $\hbar$ (that is, $\alpha=\alpha(\hbar)$),
such that $\alpha\to0$ and $\frac\hbar\alpha\to0$ as $\hbar\to0$.
For example, we can take $\alpha=C\sqrt\hbar$. The integral is understood in the weak sense: $\rho^{(\hbar)}$ is an operator such that for all $\psi,\chi\in L_2(-l,l)$ we have

$$(\psi,\rho^{(\hbar)}\chi)= \iint_\Omega\sigma(q',p')
(\psi,P[\upsilon_{q'p'}]\chi) \frac{dq'dp'}{\|\upsilon_{q'p'}\|^2}=
\iint_\Omega\sigma(q',p')
(\psi,\upsilon_{q'p'})(\upsilon_{q'p'},\chi)
\frac{dq'dp'}{\|\upsilon_{q'p'}\|^2}.$$ We claim that
$\rho^{(\hbar)}$
is a density operator, that is, a positive operator with unit trace. Indeed, the positivity of $\rho^{(\hbar)}$ follows directly from the non-negativity of $\sigma$. To prove that $\rho^{(\hbar)}$
is a trace class operator and to calculate its trace, let us take an arbitrary orthonormal basis $u_i$, $i=1,2,\ldots$, in
$L_2(-l,l)$ and calculate
\begin{multline}\label{49} 
\sum_{i=1}^\infty(u_i,\rho^{(\hbar)}u_i)=
\sum_{i=1}^\infty\iint_\Omega\sigma(q',p') |(u_i,\upsilon_{q'p'})|^2
\frac{dq'dp'}{\|\upsilon_{q'p'}\|^2}\\=
\iint_\Omega\sigma(q',p')\sum_{i=1}^\infty |(u_i,\upsilon_{q'p'})|^2
\frac{dq'dp'}{\|\upsilon_{q'p'}\|^2}=
\iint_\Omega\sigma(q',p')\,dq'dp'=1.
\end{multline}
Here the third equation follows from the Parseval--Steklov identity. To justify interchanging the sum and the integral in the second equation, we note that since the sequence $\sum_{i=1}^N|(u_i,\upsilon_{q'p'})|^2$,
$N=1,2,\ldots$, converges, it is bounded by some constant $A$. Then the integrand is majorized by the integrable function $A\sigma$ and Lebesgue theorem enables us to pass to the limit under the integral sign. Since the orthonormal basis was arbitrary, we have proved that
$\rho^{(\hbar)}$ is a density operator, as required.

For $t>0$, we have
$$\rho^{(\hbar)}_t=
\iint_\Omega\sigma(q',p') P[\upsilon_{q'p',t}]
\frac{dq'dp'}{\|\upsilon_{q'p'}\|^2},$$
$$\widetilde\rho^{(\hbar)}_t(q,p)=
\frac{1}{2\pi\hbar}\iint_\Omega\sigma(q',p')
|(\upsilon_{qp},\upsilon_{q'p',t})|^2
\frac{dq'dp'}{\|\upsilon_{q'p'}\|^2}.$$ Using formula
(\ref{44}), Theorem~\ref{PropDynAsymp} (the case $c=0$,
$D=0$) and Proposition~\ref{LemNormCirc}, we get
$$\lim[\widetilde\rho^{(\hbar)}_t(q,p)-\sigma(q-\frac{p}{m}\,t,p)]=0,$$
uniformly on $\Omega$ (and, hence, in $\mathscr D'(K)$), where the limit is performed as indicated in part 2) of the theorem. This proves parts 1), 2) of the theorem in the case when $\sigma\in\mathscr S(K)$  (part 1) is obtained by putting $t=0$). Parts 1), 2) for an arbitrary function  $\sigma\in\mathscr D'(K)$ follow since
$\mathscr S$ is dense in the set of locally integrable functions (that is, regular distributions), and the set of regular distributions is dense in the space
$\mathscr D$ of all distributions \cite{26}.

Namely, to obtain an arbitrary distribution
$\sigma\in\mathscr D'_1(K)$ as a limiting case of a Husimi functions, we take a sequence of functions
$\sigma_r\in\mathscr S(K)$, tendingto $\sigma$ in $\mathscr D(K)$
as $r\to\infty$ and put
\begin{equation}\label{50} 
\rho^{(\hbar)}=\iint_\Omega\sigma_{r(\hbar)}(q',p')
P[\upsilon_{q'p'}] \frac{dq'dp'}{\|\upsilon_{q'p'}\|^2}.
\end{equation}
Here $r(\hbar)$ is a function tending to infinity as
$\hbar\to0$. Then the limit of the Husimi transforms is the desired distribution $\sigma$.

For a distribution of the form $\sigma(q,p)=\delta(q-q_0,p-p_0)$,
$(q_0,p_0)\in\Omega$ there is another regularization, which is more straightforward than the general one (just described). We put
\begin{equation}\label{51} 
\rho^{(\hbar)}=P[\frac{\upsilon_{q_0p_0}}{\|\upsilon_{q_0p_0}\|}].
\end{equation}
Then
$\rho^{(\hbar)}_t=\rho^{(\hbar)}=P[\frac{\upsilon_{q_0p_0,t}}{\|\upsilon_{q_0p_0}\|}]$.
By Theorem~\ref{PropDynAsymp} (the case $c=0$, $D=0$) we have
\begin{equation}\label{52} 
\widetilde\rho^{(\hbar)}_t(q,p)=\frac{1}{2\pi\hbar\|\upsilon_{q_0p_0}\|^2}
|(\upsilon_{qp},\upsilon_{q_0p_0,t})|^2
\to\delta(q-q_0-\frac{p}{m}t,p-p_0).\end{equation}

We now prove part 3) of the theorem. Suppose that $\sigma\in L_1(\Omega)$ and
$\rho^{(\hbar)}$ is defined (\ref{48}) or
(\ref{50}). By Theorem~\ref{PropDynAsymp} (the case
$c=0$, $D\in(0,\infty]$) we have
\begin{equation*}
\lim[\widetilde\rho^{(\hbar)}_t(q,p)-\int_{-l}^l\sigma(q-\frac
pmt+q',p)\varphi_D(q')\,dq']=0
\end{equation*}
uniformly on $\Omega$, where the limit is realized as indicated in part 3). On the other hand, we have the convergence (\ref{46}). Substituting this in the last formula, we see
that
\begin{equation}\label{53} \lim\widetilde\rho^{(\hbar)}_t(q,p)=\frac{1}{2l}\int_{-l}^l\sigma(q',p)\,dq'
\end{equation}
in $\mathscr D'(K)$. Comparing (\ref{53}) and (\ref{46}), we conclude that (\ref{47}) remains valid. The
theorem is proved.
\end{proof}

We note that formula (\ref{47}) does not hold if 
$\sigma(q,p)=\delta(q-q_0,p-p_0)$, $(q_0,p_0)\in\Omega$, and the limit is realized as in part 3) of Theorem~\ref{TheoCoherDynClCirc}. Indeed, define
$\rho^{(\hbar)}$ by formula (\ref{51}). Then
$$
\widetilde\rho^{(\hbar)}_t(q,p)=\frac{1}{2\pi\hbar\|\upsilon_{q_0p_0}\|^2}
|(\upsilon_{qp},\upsilon_{q_0p_0,t})|^2,
$$
\begin{equation*}
\lim[\widetilde\rho^{(\hbar)}_t(q,p)-\varphi_D(q-q_0-\frac{p}{m}t)\delta(p-p_0)]=0,
\end{equation*}
where $\varphi_D(q-q_0-\frac{p}{m}t)\neq\delta(q-q_0-\frac{p}{m}t)$ since $D\neq0$.

Formula (\ref{47}) with $\sigma\in L_1(\Omega)$ also fails to hold if the limit is performed in the following way:
$\hbar\to0$, $\alpha\to0$, $\frac\hbar\alpha\to0$, $t\to\infty$,
$t-\frac MNT_{rev}\to0$, where $\frac MN\neq0$ is a rational fraction and
$T_{rev}$ is defined by formula (\ref{14}). For example,
if $\frac MN=1$, then
$$
\lim\widetilde\rho^{(\hbar)}_t(q,p)=\sigma(q,p)
$$
while $\sigma(q-\frac pmt,p)$ satisfies formula (\ref{46}).

Theorem~\ref{TheoCoherDynClCirc} shows that an appropriate choice of a family of density operators enables us to obtain in the semiclassical limit any classical density function. Thus, in this sense, there are at least as many quantum states as classical states. This density function evolves in accordance with the laws of classical mechanics at time scale $T_{cl}$ independently of whether it is a delta-function or an integrable function. 

A difference appears at time scale $T_{coll}$. We have seen that if $\sigma$ is a delta-function, then the classical dynamics deviates from the quantum dynamics because a quantum packet collapses while a classical point-like particle remains the point-like particle for all time. Thus, the Newtonian mechanics holds only at time scale $T_{cl}$. From the other side, if $\sigma$ is an integrable function, then Theorem~\ref{TheoCoherDynClCirc} shows that the classical dynamics remains valid from the quantum-mechanical point of view even at time scale $T_{coll}$ since it exhibits a collapse of spatial density just as the quantum dynamics does.

In both cases, the classical dynamics deviates from the quantum dynamics at time scale $T_{rev}$ because of the purely quantum effect of fractional revivals of packets.

Thus, functional classical mechanics remains valid at a larger time scale than the Newtonian classical mechanics. Thus, it is preferable from the quantum-mechanical point of view.

\subsection{The case of quantum dynamics in a box}Since the family $\omega_{qp}$, $(q,p)\in\Omega$ is a resolution of unity in $L_2(-l, l)$, the Husimi transform may be defined as
\begin{equation*}
\widetilde\rho(q,p)=\frac{1}{2\pi\hbar}\Tr P[\omega_{qp}]\rho=
\frac{1}{2\pi\hbar}(\omega_{qp},\rho\,\omega_{qp}),
\end{equation*}
where $\rho$ is a density operator (a quantum state) in
$L_2(-l,l)$ and $\widetilde\rho(q,p)$ is a density function on the phase space  $\Omega$ (a classical state). If $\rho=P[\psi]$ for some
$\psi\in L_2(-l,l)$, then
\begin{equation*}
\widetilde\rho(q,p)=\frac{1}{2\pi\hbar}|(\psi,\omega_{qp})|^2.
\end{equation*}
The Husimi transform based on the functions $\omega_{qp}$ is suitable for studying the correspondence between the classical and quantum mechanics for a particle in a box. We pose the same questions as in the case of a particle on a circle: is it true that every classical state can be obtained as a semiclassical limit of Husimi functions of quantum states, and at which time scale is this correspondence preserved for various classical states?

A difference of the box case from the circle case is that $\omega_{\pm
l,0}\equiv0$, whence we have $\widetilde\rho(\pm l,0)=0$. Therefore, we restrict ourselves to those density functions $\sigma$ which vanish in some neighbourhoods of the points $(\pm l,0)$. This restriction is physically inessential since the neighbourhoods can be arbitrarily small.

Let the quantum system evolve in time in the box:
$\rho_t=U^b_t\rho U_t^{b\dag}$, where
$U^b_t$ is the evolution operator in a box as defined by formula (\ref{38}) and $\rho$ is an initial density operator. We again denote the Husimi function of $\rho_t$ by $\widetilde\rho_t(q,p)$.

We can now prove an analogue of Theorem~\ref{TheoCoherDynClCirc} for dynamics in a box. The non-negativity and unit normalization of distributions in $\mathscr D'(\Omega)$ is defined as in $\mathscr D'(K)$. We denote the set of all non-negative and unit-normalized distributions in $\mathscr D'(\Omega)$ by $\mathscr D'_1(\Omega)$.

\begin{theorem}\label{TheoCoherDynClSegm}
1) Let $\sigma\in\mathscr D'_1(\Omega)$ and the support of $\sigma$ is disjoint from some neighbourhoods of the points $(\pm l,0)$. Then there is a family of density operators $\rho^{(\hbar)}$,
$\hbar>0$, in $L_2(-l,l)$, such that their Husimi functions $\widetilde\rho^{(\hbar)}$ converge to $\sigma$ in $\mathscr
D'(\Omega)$ as $\hbar\to0$, $\alpha\to0$, $\frac\hbar\alpha\to0$;

2) Let, further, $\rho^{(\hbar)}_t=U^b_t\rho^{(\hbar)}U_t^{b\dag}$,
$\widetilde\rho^{(\hbar)}_t$ and let $\rho^{(\hbar)}_t$
be the Husimi function of $\rho^{(\hbar)}_t$. Then
\begin{equation}\label{54} 
\lim[\widetilde\rho^{(\hbar)}_t(q,p,t)-
\sigma(q-\frac{p}{m}t,p)]=0\quad\text{in } \mathscr
D'(\Omega).\end{equation} Here the limit is performed as follows: $\hbar\to0$, $\alpha\to0$, $\frac\hbar\alpha\to0$,
$t=t(\hbar)$, $\frac{\hbar t}\alpha\to0$;

3) If $\sigma\in L_1(\Omega)$, then equality
(\ref{54}) remains valid in the following limit: $\hbar\to0$, $\alpha\to0$,
$\frac\hbar\alpha\to0$, $t\to\infty$, $\hbar t\to0$, $\frac{\hbar
t}\alpha\to2mD\in(0,\infty]$.
\end{theorem}

The proof is similar to the proof of Theorem~\ref{TheoCoherDynClCirc} but we now use Theorem~\ref{PropDynAsymp} instead of Theorem~\ref{PropDynAsymp}. In particular, if $\sigma\in\mathscr S(\Omega)$, then
$$\rho^{(\hbar)}=\iint_\Omega\sigma(q',p')P[\omega_{q'p'}]
\frac{dq'dp'}{\|\omega_{q'p'}\|^2},$$ and every distribution from $\mathscr D'(\Omega)$ can be approximated by functions from $\mathscr
S(\Omega)$. Note that the integral on the right-hand side is well defined even though $\omega_{\pm l,0}=0$. This is because we assumed that $\sigma$ vanishes in some neighbourhoods of the points
$(\pm l,0)$. But even if we take an arbitrary $\sigma\in\mathscr S(\Omega)$,
the points $(\pm l,0)$ are not poles for the integrand by the Cauchy-–Bunyakovsky inequality:
$$\frac{1}{2\pi\hbar}\iint_\Omega
\frac{|(\omega_{qp},\omega_{q'p',t})|^2}{\|\omega_{q'p'}\|^2}\,
\sigma(q',p')\,dq'dp'\leq\frac{\|\omega_{qp}\|^2}{2\pi\hbar}
\iint_\Omega\sigma(q',p')\,dq'dp'.$$

If $\sigma$ has the form $\sigma(q,p)=\delta(q-q_0,p-p_0)$ for some
$(q_0,p_0)\neq(\pm l,0)$, then we can propose a special approximating family of operators
$\rho^{(\hbar)}=P[\frac{\omega_{q_0p_0}}{\|\omega_{q_0p_0}\|}]$.

Thus, our conclusions for dynamics in a box are the same as for dynamics on a circle. The classical dynamics of a material point and the dynamics of a density function are both obtained in the limit $\hbar\to 0$ at time scale $T_{cl}$. But the Newtonian fails to describe the collapse of the probability density and, hence, ceases to be valid at time scale $T_{coll}$. Functional mechanics, as well as quantum mechanics, exhibits the collapse of the probability density. But it does not describe the revival of wave packets, which is, thus, a purely quantum phenomenon having no analogue in classical mechanics. Hence, functional mechanics remains valid at time scale $T_{coll}$ but not at time scale $T_{rev}$. We see that  functional mechanics remains valid at a longer time scale than Newtonian one and, therefore, it is preferable from the quantum-mechanical point of view.

\begin{remark}\label{r8}In classical mechanics, analogues of the theorems on diffusion are proved for systems of interacting particles \cite{35}. Therefore, it would be useful to generalize the  obtained results to quantum dynamics of interacting particles on a circle and in a box, as mentioned in Remark~\ref{r6}.
\end{remark}

\section{Quantum dynamics in a box of size known with a random error}\label{s5}
We have so far assumed that the length $2l$ of the box is known with the infinite accuracy. But one of the basic postulates of functional mechanics says that the parameters of a physical system cannot be measured with the infinite accuracy. Hence, if we wish to predict the probability that the particle is in some place at a given time, we must replace the fixed value $l \in \mathbb R$ by the density function $f(l)$ of some probability distribution of the parameter $l$. In particular, one can construct this function from the results of measurements \cite{17}. The dispersion of such a distribution may be very small, but it is always non-zero. Since $T_{rev}$ depends on $l$, this breaks the exact periodicity of free quantum dynamics. Moreover, it turns out that the spatial density in a finite volume has a limit as $t \to \infty$.

\begin{theorem}Let $\psi_l\in L_2(-l,l)$, $l>0$, be any family of state vectors such that $|\psi_l(x)|^2$ is integrable with respect to $(x,l)$ on $[-l,l]\times[0,\infty)$, and let $f(l)$ be a continuous probability distribution density supported on $[0,\infty)$ (that is, $f(l)\geq0$, $\int_0^\infty f(l)dl=1$). Moreover, let $f(l)=o(l^\varepsilon)$, $\varepsilon>0$, as $l\to0$). Then, there is a limit
$$\lim_{t\to\pm\infty}P(x,t)=P_\infty(x)=\int_0^\infty\frac{\chi_l(x)}{2l}
\left[1-\frac12(\psi_l(y),\psi_l(y+2x-2l)+\psi_l(y-2x+2l))\right]f(l)\,dl,$$
where
$$P(x,t)=\int_0^\infty\chi_l(x)|\psi_{t}(x,t)|^2f(l)\,dx,$$
$\chi_l(x)$ is the characteristic function of the interval $[-l,l]$, $\psi_l(x,t)=U^b_t\psi_l(x)$, and the functions $\psi_l$ are extended to the whole real line by the formula $\psi_l(x+2nl)=(-1)^n\psi_l[(-1)^n(x-2nl)]$, $n=\pm1,\pm2,\ldots$.
\end{theorem}
The function $P(x,t)$ is the probability distribution density at time $t$ taking account of the random error in the determination of $l$. The value of $P(x,t)$ for a given $x$ contains contributions from the functions $|\psi_{l}(x,t)|^2$ over all $l$ such that $x$ lies in the interval $[-l,l]$. This explains the factor $\chi_l(x)$ in the expression for $P(x,t)$.

Let us analyse the limit probability distribution $P_\infty(x)$. The term
$$\int_0^\infty\frac{\chi_l(x)}{2l}f(l)\,dl$$
corresponds to the uniform distribution: $\chi_l(x)/2l$ is the uniform distribution on the interval $[-l,l]$, while the factor $f(l)$ corresponds to the density of the probability that
the half-length of the interval is $l$. The additional term
$$\Delta(x)=\int_0^\infty\frac{\chi_l(x)}{4l}
(\psi_l(y),\psi_l(y+2x-2l)+\psi_l(y-2x+2l))f(l)\,dl$$
is, thus, a correction to the uniform distribution. Interestingly, it produces a dependence of the limiting final probability distribution on the initial one. We shall prove below that this correction tends to zero in the semiclassical limit. We also note that the limits as $t\to+\infty$ and $t\to-\infty$, coincide, just as in the Kozlov's theorems. This reflects the time symmetry of quantum dynamics.

\begin{proof}
Let us expand the functions $\psi_l(x,t)$ into a Fourier series
$$\psi_l(x,t)=\frac1{\sqrt l}\sum_{k=1}^\infty a_k(l)\sin\left(\frac{\pi k}{2l}(x-l)\right)
\exp\left\{-\frac{i\pi^2\hbar k^2t}{8ml^2}\right\}$$
and substitute this in the expression for $P(x,t)$:
\begin{multline*}P(x,t)=\int_0^\infty\frac{\chi_l(x)}l\sum_{k,n=1}^\infty a_k\overline{a_n}
\sin\left(\frac{\pi k}{2l}(x-l)\right)\sin\left(\frac{\pi n}{2l}(x-l)\right)\\\times\exp\left\{-\frac{i\pi^2\hbar (k^2-n^2)t}{8ml^2}\right\}
f(l)\,dl.\end{multline*}
By the Riemann--Lebesgue lemma, all terms with $k\neq n$ tend to zero as $t\to\pm\infty$. Hence,
$$\lim_{t\to\pm\infty}P(x,t)=P_\infty(x)=
\int_0^\infty\frac{\chi_l(x)}l\sum_{k=1}^\infty
|a_k|^2\sin^2\left(\frac{\pi k}{2l}(x-l)\right) f(l)\,dl.$$ The theorem now follows from the formula
\begin{multline*}
\sum_{k=1}^\infty|a_k|^2\sin^2\left(\frac{\pi k}{2l}(x-l)\right)=
\frac12\sum_{k=1}^\infty |a_k|^2\left[1-\cos\left(\frac{\pi k}{2l}(x-l)\right)\right]\\=
\frac12\left[1-\frac12(\psi_l(y),\psi_l(y+2x-2l)+\psi_l(y-2x+2l))\right].
\end{multline*}
\end{proof}

This result is rather unexpected since the quantum dynamics in a bounded domain is commonly regarded as being almost periodic \cite{1} and, hence, has no limit for large values of time.

We have already mentioned that the limit distribution differs from the uniform one. To estimate their difference $\Delta(x)$ in the semiclassical limit, we first assume that $\psi_l = \omega_{qp}$ is a coherent state (we omit the subscript $l$) and then represent an arbitrary state by an integral over coherent states.

\begin{proposition}
Let $\psi_l=\omega_{qp}$ for all $l>0$. Consider the semiclassical limit  $\hbar,\alpha,\frac\hbar\alpha\to0$. Then $\Delta(x)$ tends to zero in the weak sense, that is, for every
element $\sigma$ of the space $\mathscr S(\mathbb R)$ of rapidly decaying functions we have
$$\int_{-\infty}^{+\infty}\Delta(x)\sigma(x)\,dx\to0.$$
\end{proposition}
\begin{proof}
Expanding $\sigma$ into the Fourier integral
$$\sigma(x)=\frac1{\sqrt{2\pi}}\int_{-\infty}^{+\infty}\widetilde\sigma(\lambda)e^{i\lambda x}d\lambda$$
and using the formulae $\omega_{qp}(y+2x-2l)=\omega_{-q-2x,-p}(y)$ and $\omega_{qp}(y-2x+2l)=\omega_{-q+2x,-p}(y)$, we get
$$\int_{-\infty}^{+\infty}\Delta(x)\sigma(x)\,dx=
\frac1{4\sqrt{2\pi}l}\int_{-\infty}^{+\infty}d\lambda\,\widetilde\sigma(\lambda)
\int_0^\infty dl\,f(l)\int_{-\infty}^{+\infty}
(\omega_{qp},\omega_{-q-2x,-p}+\omega_{-q+2x,-p})\,dx.$$ 
Calculating the integral over $x$ (the calculation is analogous to those in Section~\ref{s3}), we see that this expression tends to zero. The proposition is proved.
\end{proof}
Consider again arbitrary functions  $\psi_l\in L_2(-l,l)$, $l>0$, and express them as integrals over coherent states (a corollary of (\ref{36})):
$$\psi_l=\frac1{2\pi\hbar}\iint_\Omega g(q,p)\,\omega_{qp}\,dqdp,$$
where $g(q,p)=(\omega_{qp},\psi_l)$. We perform the semiclassical limit in the following way. As usual, let $\hbar,\alpha,\frac\hbar\alpha$ tend to zero as parameters of the state $\omega_{qp}$, but let the function $g$ remains constant. In other word, the dependence of $\psi_l$ on $\hbar$ and $\alpha$ is such that the scalar product $(\omega_{qp},\psi_l)$ coincides with the fixed function $g(q,p)$.

\begin{corollary}
Under this semiclassical limit of arbitrary functions $\psi_l$, we have $\Delta(x)\to0$ in the weak sense.
\end{corollary}
The proof can be performed by a direct calculation.

We recall that the weak convergence of probability density functions is sufficient from a physical point of view \cite{Kozlov3,11,12}.

Thus, taking account of the inevitable random error in measuring the size of the box, we see that the probability distribution of the position of a particle in a box has a limit as $t \to\infty$, and the limiting distribution is non-uniform and depends on the initial one. But in the semiclassical approximation, the limiting distribution becomes uniform.

\section*{Acknowledgements}
The authors are grateful for useful remarks and discussions to M.\,V.~Berry, S.\,V.~Bolotin, S.\,Yu.~Dobrokhotov, J.\,R.~Klauder, V.\,V.~Kozlov, V.\,I.~Man'ko, A.\,G.~Sergeev, O.\,G.~Smolyanov, A.\,D.~Sukhanov, D.\,V.~Treshchev, B.\,L.~Voronov, V.\,V.~Vedenyapin, and E.\,I.~Zelenov. This work was partially supported by the Russian Foundation for Basic Research (project 11-01-00828-a), the Russian Federation's President Programme for the Support of Leading Scientific Schools (project NSh-2928.2012.1), and the Programme of the Division of Mathematics of the Russian Academy of Sciences.

\appendix

\section*{Appendix}

Here we make a remark on the modular property of the theta-function

\begin{equation*}
\theta(z,\tau)=\sum_{k=-\infty}^{+\infty}\exp\{-\pi\tau k^2+2\pi
ikz\},\quad z,\tau\in\mathbb C,\quad \Re\tau>0,
\end{equation*}
in a form applicable to the properties of quantum coherent states on an interval.

We have the well-known modular relation
\begin{equation}\label{55} 
\theta\left(\frac{z}{i\tau},\frac{1}{\tau}\right)=\sqrt\tau e^{\frac{\pi
z^2}{\tau}}\theta(z,\tau).
\end{equation}

To see this, rewrite it in the form
\begin{equation}\label{56} 
\frac1{\sqrt\tau}\sum_{n=-\infty}^{+\infty}\exp\left\{-\frac{\pi(z-n)^2}{\tau}\right\}=
\sum_{k=-\infty}^{+\infty}\exp\{-\pi\tau k^2+2\pi ikz\}
\end{equation}

Express $z\in\mathbb C$ in the form $z=x+iy$, $x,y\in\mathbb R$. The left-hand side of  (\ref{56}) belongs to $L_2(0,1)$ as a function of $x$ for fixed $y$ and $\tau$.
Hence, it can be expanded into a Fourier series with respect to the orthonormal basis $\{e^{i\pi kx},\: k=0,\pm1,\pm2,\ldots\}$:
\begin{equation}\label{57} 
\frac1{\sqrt\tau}\sum_{n=-\infty}^{+\infty}\exp\left\{-\frac{\pi(x+iy-n)^2}{\tau}\right\}=
\sum_{k=-\infty}^{+\infty}a_ke^{i\pi kx}.
\end{equation} 
We find the Fourier coefficients in the following way:
\begin{multline*}
a_k=\frac1{\sqrt{\tau}}\sum_{n=-\infty}^{+\infty}\int_0^1
\exp\left\{-\frac{\pi(x+iy-n)^2}{\tau}-i\pi kx\right\}dx\\=
\frac1{\sqrt{\tau}}\int_{-\infty}^{+\infty}
\exp\left\{-\frac{\pi(x+iy)^2}{\tau}-i\pi kx\right\}dx=\exp\{-\pi\tau k^2+2\pi i
k(iy)\}.
\end{multline*}
Substituting this expression in (\ref{57}), we obtain (\ref{56}).

Thus, the modular property of the theta-function can be viewed as a rephrase of the Fourier series
expansion of the function which is the sum of the Gaussian functions centred at the integer points.

\end{document}